\documentclass[a4paper,UKenglish,cleveref,numberwithinsect,final]{lipics-v2019}

\usepackage[inline, marginclue]{fixme}
\fxusetheme{color}
\FXRegisterAuthor{ls}{als}{LS}
\FXRegisterAuthor{dp}{atm}{DP}
\FXRegisterAuthor{mg}{amg}{MG}
\usepackage{acronym}[2015/03/22]
\usepackage{xspace}

\usepackage{stmaryrd}
\usepackage[vlined,ruled]{algorithm2e}
\usepackage{hyperref}
\hypersetup{linkcolor = blue}

\theoremstyle{definition}
\newtheorem{notn}[theorem]{Notation}
\newtheorem{defn}[theorem]{Definition}

\newcommand*{\rs}[0]{\ensuremath{(\mathit{DES}_0)}\xspace}
\newcommand*{\rc}[0]{\ensuremath{(\mathit{DES}_1)}\xspace}
\newcommand{\col}{\mathsf{col}}
\newcommand{\At}{\mathsf{At}}

\newcommand{\inda}{I}
\newcommand{\indb}{J}

\newcommand{\Rules}{\mathcal{R}}
\newcommand{\omod}{\mathcal{G}}
\newcommand{\Ag}{\mathit{Ag}}
\newcommand{\bbrack}[1]{{\llbracket #1 \rrbracket}}
\newcommand{\PV}{\mathsf{V}}
\newcommand{\Sem}[1]{\bbrack{#1}}
\newcommand{\hearts}{\mathop{\heartsuit}}
\newcommand{\F}{F}
\newcommand{\Gp}{\mathsf{G}}
\newcommand{\GP}{\Gp_\mathsf{p}}
\newcommand{\GES}{\Gp_\mathsf{ES}}
\newcommand{\Pow}{\mathcal{P}}
\newcommand{\Nat}{\mathbb{N}}
\newcommand{\Int}{\mathbb{Z}}
\newcommand*{\MC}[0]{\(\mu\)-calculus\xspace} 
\newcommand{\ATLdiamond}[1]{\mathop{\langle\!\langle#1\rangle\!\rangle}}

\newcommand{\CT}{\ensuremath{\mathcal{T}}}

\makeatletter 
\newcommand*{\org@overidelabel}{}
\let\org@overridelabel\@verridelabel
\renewcommand*{\@verridelabel}[1]{%
  \@bsphack
  \protected@write\@auxout{}{\string\AC@undonewlabel{#1@cref}}%
  \org@overridelabel{#1}%
  \@esphack
}%

\makeatother

\acrodef{CL}{coalition logic}
\acrodef{CLES}{\ac{CL} with explicit strategies}
\acrodef{CLDES}{\ac{CL} with disjunctive explicit strategies}
\acrodef{CGF}{concurrent game frame}
\acrodef{CGS}{concurrent game structure}
\acrodef{CGSES}{concurrent game structure with explicit strategies}
\acrodef{ATL}{alternating-time temporal logic}
\acrodef{ATEL}{alternating-time temporal epistemic logic}
\acrodef{MAS}{multi-agent system}
\acrodef{CTL}{computation tree logic}
\acrodef{INL}{instantial neighborhood logic}
\acrodef{PDL}{propositional dynamic logic of regular programs}
\acrodef{IPDL}{instantial \ac{PDL}}
\acrodef{CATL}{counterfactual \ac{ATL}}
\acrodef{ATLA}[ATL-A]{\ac{ATL} with actions}
\acrodef{DGL}{dynamic game logic}
\acrodef{SDGL}{strategized \ac{DGL}}
\acrodef{ATLES}{\ac{ATL} with explicit strategies}
\acrodef{ATLDES}{\ac{ATL} with disjunctive explicit strategies}
\acrodef{ATLEA}{\ac{ATL} with explicit actions}
\acrodef{AL}{action logic}
\acrodef{EL}{environment logic}
\acrodef{BML}{boolean modal logic}
\acrodef{AMC}{alternating-time \MC}
\acrodef{AMCES}{\ac{AMC} with explicit strategies}
\acrodef{AMCDES}{alternating-time \MC with disjunctive explicit strategies}

\newcommand*{\defemph}[1]{\emph{#1}} 

\newcommand*{\moves}[1]{\ensuremath{[k#1]}}
\newcommand*{\Agents}[0]{\ensuremath{\Sigma}}
\newcommand*{\interpret}[0]{\ensuremath{\iota}}

\newcommand{\cldiamond}[1]{\ensuremath{\mathop{\langle #1 \rangle}}}
\newcommand{\clbox}[1]{\ensuremath{\mathop{[#1]}}}
\renewcommand*{\implies}{\ensuremath{\rightarrow}}
\newcommand{\defeq}[0]{\ensuremath{:=}}
\newcommand{\mcauses}[0]{\ensuremath{\mathrel{\Diamond\kern-.3em\raise.1em\hbox{\(\implies\)}}}}

\newcommand*{\rng}[1]{\ensuremath{1, \dots, #1}}
\newcommand*{\srng}[1]{\ensuremath{\{\rng{#1}\}}}

\newcommand*{\ET}[0]{\textsc{ExpTime}\xspace}
\newcommand*{\PS}[0]{\textsc{PSpace}\xspace}
\newcommand*{\PT}[0]{\textsc{P}\xspace}
\newcommand*{\QP}[0]{\textsc{QP}\xspace}
\newcommand*{\NP}[0]{\textsc{NP}\xspace}
\newcommand*{\PTIME}[0]{\textsc{P}\xspace}
\newcommand*{\coNP}[0]{\textsc{coNP}\xspace}

\newcommand*{\wrt}[0]{w.r.t.\ }
\newcommand*{\ie}[0]{i.e.\ }
\newcommand*{\eg}[0]{e.g.\ }

\crefname{algocf}{algorithm}{algorithms}
\crefname{lem}{Lemma}{Lemmas}

\title{The Alternating-Time $\mu$-Calculus With Disjunctive Explicit Strategies}

\author{Merlin Humml}{Friedrich-Alexander-Universität Erlangen-Nürnberg, Erlangen, Germany}{merlin.humml@fau.de}{https://orcid.org/0000-0002-2251-8519}{Work performed under the DFG Project \emph{Reconstructing Arguments from Noisy Text (RANT)}, SCHR 1118/14-1 / SCHR 1118/14-2}
\author{Lutz Schröder}{Friedrich-Alexander-Universität Erlangen-Nürnberg, Erlangen, Germany}{lutz.schroeder@fau.de}{http://orcid.org/0000-0002-3146-5906}{Work performed under the DFG Project \emph{RANT}, SCHR 1118/14-1 / SCHR 1118/14-2}
\author{Dirk Pattinson}{The Australian National University, Canberra,  Australia}{dirk.pattinson@anu.edu.au}{https://orcid.org/0000-0002-5832-6666}{}

\authorrunning{M. Humml, L. Schröder, and D. Pattinson}
\Copyright{Merlin Humml, Lutz Schröder, and Dirk Pattinson}



\category{} 

\relatedversion{} 

\EventEditors{Christel Baier and Jean Goubault-Larrecq}
\EventNoEds{2}
\EventLongTitle{29th EACSL Annual Conference on Computer Science Logic (CSL 2021)}
\EventShortTitle{CSL 2021}
\EventAcronym{CSL}
\EventYear{2021}
\EventDate{January 25--28, 2021}
\EventLocation{Ljubljana, Slovenia (Virtual Conference)}
\EventLogo{}
\SeriesVolume{183}
\ArticleNo{14}

\begin{document}
\nolinenumbers
\maketitle

\begin{abstract}
  \Ac{ATL} and its extensions, including the \ac{AMC}, serve the
  specification of the strategic abilities of coalitions of agents in
  concurrent game structures.  The key ingredient of the logic are
  path quantifiers specifying that some coalition of agents has a
  joint strategy to enforce a given goal.  This basic setup has been
  extended to let some of the agents (revocably) commit to using
  certain named strategies, as in \emph{\ac{ATLES}}.  In the present
  work, we extend \ac{ATLES} with fixpoint operators and strategy
  disjunction, arriving at the \emph{\ac{AMCDES}}, which allows for a
  more flexible formulation of temporal properties (e.g.\ fairness)
  and, through strategy disjunction, a form of controlled
  non-determinism in commitments. Our main result is an \(\ET\) upper
  bound for satisfiability checking (which is thus
  \(\ET\)-complete). We also prove upper bounds \(\QP\)
  (quasipolynomial time) and \(\NP\cap\coNP\) for model checking under
  fixed interpretations of explicit strategies, and \(\NP\) under open
  interpretation. Our key technical tool is a treatment of the
  \ac{AMCDES} within the generic framework of coalgebraic logic, which
  in particular reduces the analysis of most reasoning tasks to the
  treatment of a very simple \emph{one-step logic} featuring only
  propositional operators and next-step operators without nesting; we
  give a new model construction principle for this one-step logic that
  relies on a set-valued variant of first-order resolution.
\end{abstract}

\acresetall{} 
\clearpage
\section{Introduction}\label{sec:introduction}

\emph{\Ac{ATL}}~\cite{AlurEA02} extends \ac{CTL} with path quantifiers
$\ATLdiamond{A}$ read `coalition~$A$ of agents has a (long-term) joint
strategy to enforce'.  It is embedded into the \emph{\ac{AMC}}, which
instead of path quantifiers, features nested least and greatest
fixpoints alongside the next-step coalition modalities
$\ATLdiamond{A}\bigcirc$ (`$A$ can enforce in the next
step'). 
The \ac{AMC} is strictly more
expressive than \ac{ATL}, e.g.\ supports fairness
constraints.

Coalitional power in~\ac{ATL} and the~\ac{AMC} is measured without any
restrictions on the moves chosen by the opponents. There has been
interest in extensions of \ac{ATL} where the power of the opponents
can be constrained, e.g.\ by committing some of them to a particular
strategy, allowing for statements such as `no matter what the other
network actors do, Alice and Bob can collaborate to exchange keys via
Server~$S$ provided that~$S$ adheres to the protocol'. One such
extension is provided in \emph{\ac{ATLES}}~\cite{WaltherEA07}, which
has path quantifiers $\ATLdiamond{A}_\rho$ additionally parametrized
over a commitment~$\rho$ of some agents to given named
strategies, 
read `provided that the commitments~$\rho$ are kept,~$A$ can enforce
\dots'. This extension has substantial impact on expressiveness; e.g.\
unlike in basic \ac{ATL}, the semantics of \ac{ATLES} over
history-free strategies differs from the one over history-dependent
strategies.

Restricting opponents to fixed moves is, of course, quite drastic; as
noted already in the conclusion of Walther~\cite[Chapter~4]{Walther_PhD}, it
is desirable to allow for more permissive restrictions where the
opponents can still pick among several designated moves, as in `Alice
has a strategy to get her print job executed if Bob either cancels his
large print job or splits it into several smaller ones'. In the
present paper, we introduce such an extension with disjunctive
commitments. Additionally, we include full support for least and
greatest fixpoint operators, with associated gains in expressivity
analogous to the extension from \ac{ATL} to the \ac{AMC}.  We thus
arrive at the \emph{\ac{AMCDES}}.

Our main result on this logic is that satisfiability checking remains
only $\ET$-complete (i.e.\ no harder than the \ac{AMC}, or in fact
than basic \ac{ATL} or even \ac{CTL}). We note also that (following a
distinction made also in work on \ac{ATLES}~\cite{WaltherEA07}) model
checking is in quasipolynomial time $\QP$ and in $\NP\cap\coNP$ under
fixed interpretation of explicit strategies (matching the best known
bounds for the~\ac{AMC} and in fact even the plain relational
$\mu$-calculus), and in~$\NP$ under open interpretation; these results
are obtained by fairly straightforward adaptation of results on the
\ac{AMC}~\cite{DBLP:conf/concur/HausmannS19}, and therefore discussed
in full only in the appendix.  We obtain our results by casting the
\ac{AMCDES} as an instance of \emph{coalgebraic
  logic}~\cite{CirsteaEA11a}, a unifying framework for modal and
temporal logics. The driving principle of coalgebraic logic is to
reduce reasoning tasks to the analysis of a simple \emph{one-step
  logic}, whose formulae employ only Boolean connectives and a single
layer of next-step
modalities~\cite{SchroderPattinson09,CirsteaEA11b,DBLP:conf/concur/HausmannS19}. In
particular, the automata- and game-theoretic machinery needed for the
treatment of fixpoint logics is entirely encapsulated in results on
the coalgebraic $\mu$-calculus~\cite{CirsteaEA11b,DBLP:conf/concur/HausmannS19}.
The actual technical work then lies in providing algorithms,
axiomatizations, and model constructions for the one-step logic of
\ac{AMCDES}, still posing substantial challenges due to nested
quantification over strategies. The model construction principle for
the one-step logic that we employ is based on a set-valued variant of
first-order resolution that we introduce here, along with an
associated notion of equationally complete model that we use to move
from (generally infinite) Herbrand universes to finite models; this
principle is the key to supporting strategy disjunction.

The present material revises and extends a previous conference
publication~\cite{GottlingerEA21}.


\subparagraph*{Related Work}\label{sec:related-work}
Many \ac{ATL} extensions are concerned with commitments of agents to
strategies. Besides \acf{ATLES}, this includes, e.g.,
\acl{CATL}~\cite{van_der_Hoek_2005}, which differs from \ac{ATLES} by
making commitments irrevocable. \acf{ATLA}~\cite{Agotnes06} has
per-agent disjunctive commitments (while the \ac{AMCDES} allows
disjunctions over joint commitments). \ac{ATLA} admits polynomial-time
model checking; satisfiability checking is not considered (it would be
somewhat simpler than in the present setting, as in \ac{ATLA} all
actions are named, and hence known in
advance). \ac{ATLEA}~\cite{Herzig_2013} features commitments of agents
to a given action at only the current world, and has a fairly
straightforward satisfiability-preserving embedding into
the~\ac{AMCDES}.  Various forms of \emph{strategy
  logic}~\cite{ChatterjeeEA10,DBLP:journals/tocl/MogaveroMPV14,DBLP:journals/corr/MogaveroMPV16}
possibly contain \ac{ATL}$^*$ with disjunctive explicit strategies
(but presumably not the~\ac{AMCDES} or even the~\ac{AMC}, as they lack
fixpoint operators); they tend to be computationally much harder than
the~\ac{AMCDES}. Goranko and Ju~\cite{DBLP:conf/lori/GorankoJ19}
discuss various forms of conditional strategic modalities, one of
which ($O_{dd}$) is similar in spirit to our strategy disjunction in
that it restricts the moves of the opposition, however not to given
named moves but rather to moves enforcing a given goal; their main
technical result is a Hennessy-Milner style expressiveness theorem. De
Nicola and Vandraager~\cite{{DBLP:conf/litp/NicolaV90}} consider
disjunction of named actions in labelled transition systems, which in
that setting can be encoded into next-modalities for single actions
using logical disjunction.

\subparagraph*{Organization} We introduce the syntax and the semantics of
the \acf{AMCDES} in~\cref{sec:amcdes}. After recalling the requisite
principles of coalgebraic logic in~\cref{sec:coalg-log} we introduce the method of
set-valued first-order resolution in~\cref{sec:set-valued-resolution}.
We illustrate these methods on the
basic \ac{AMC} in \cref{sec:coalg-amc}, and establish our main results
on satisfiability checking for the \ac{AMCDES} in
\cref{sec:amcdes-sat}.

\section{\Ac{AMC} With Disjunctive Explicit Strategies}\label{sec:amcdes}

We proceed to introduce the syntax and semantics of the
\acf{AMCDES}. As indicated in \cref{sec:introduction}, the
logic is inspired by \acf{ATLES}~\cite{WaltherEA07}. We deviate from
the \ac{ATLES} syntax in that we express (disjunctive) commitments of
agents by means of names for strategies in the syntax. Also, we
shorten the \ac{ATL} syntax for next-step operators from
$\ATLdiamond{C}\bigcirc$ (`$C$ can enforce in the next step that
\dots') to $\clbox{C}$ as in coalition logic~\cite{Pauly02}. We thus
arrive at modalities $\clbox{C,O}$ where $O$ is a set of named joint
strategies for agents in a further coalition~$D$ of agents restricted in their choice of strategies, disjoint from~$C$,
read `if the agents in~$D$ use one of the joint strategies in~$O$,
then $C$ can enforce that \dots'. The dual modality $\cldiamond{C,O}$
is read `even if the agents in~$D$ are limited to the joint strategies
in~$O$, $C$ cannot prevent that \dots'. Formally, our syntax is
defined as follows.

\newcommand*{\ieq}[0]{\mathrel{=_\sqcap}}
\newcommand{\other}[1]{\ensuremath{\overline{#1}}}
\begin{defn}
  \label{def:syntax-amcdes}
  The syntax of the \ac{AMCDES} is parametrized over a set \(\At\) of
  (propositional) \defemph{atoms}, \(V\) of \defemph{variables}, a
  finite set \(\Agents\) of agents (for technical simplicity, assumed
  to be linearly ordered), and sets \(M_j\) of \defemph{explicit
    strategies} (\ie names for strategies) per agent \(j\); \emph{we
    fix these data from now on}. A \emph{coalition} is a subset
  of~$\Agents$. We also (and mainly) refer to explicit strategies as
  \emph{explicit moves}. We write $M_D=\prod_{j\in D}M_j$ for the set
  of \emph{joint explicit moves} of a coalition~$D$.  Formulae
  $\phi,\psi$ are then given by the grammar
  \begin{equation*}
    \phi, \psi ::= p \mid \neg p \mid x \mid \top \mid \bot\mid \phi \land \psi\mid \phi \lor \psi \mid \clbox{C, O}\phi \mid \cldiamond{C, O}\phi 
    \mid \mu x .\, \phi\mid \nu x .\, \phi
  \end{equation*}
  where \(x \in V\), \(p \in \At\), and $C \subseteq \Agents$,
  i.e.\ a \defemph{coalition}. We generally write
  $\other{C}=\Agents\setminus C$. Moreover, $O \subseteq M_D$ is a
  set of joint explicit moves, called a \defemph{disjunctive explicit
    strategy} (or \defemph{move}), for some coalition~$D$, disjoint
  from~$C$, that we denote by \(\Ag(O)\). We call a modality
  $\clbox{C,O}$ or $\cldiamond{C,O}$ a \emph{grand coalition modality}
  if $C\cup\Ag(O)=\Agents$, and \emph{non-disjunctive} if $|O|=1$, in
  which case we often omit set brackets and just write~$O$ as its
  single element. We restrict grand coalition modalities to be
  non-disjunctive (cf.\ \cref{rem:grand}). As usual, $\mu$ and~$\nu$
  take least and greatest fixpoints, respectively. Negation~$\neg$ is
  not included but can be defined in the standard way, taking negation
  normal forms. The \ac{AMCES} is the fragment of the \ac{AMCDES}
  allowing only non-disjunctive modalities.
\end{defn}
The \ac{AMCDES} thus subsumes both the standard
\ac{AMC}~\cite{AlurEA02} (with~$\clbox{C}$ corresponding to
$\clbox{C,O}$ with $\Ag(O)=\emptyset$) and the history-free variant of
\ac{ATLES}~\cite{WaltherEA07} (which as we will detail in \cref{rem:hist} is
the variant to which previous technical results refer).
\begin{example}\label{expl:formulae}
  The formula indicated in the introduction,
 \begin{equation*}
   \clbox{\mathsf{Alice},(\mathsf{Bob}\colon\{\mathsf{cancelPrint},\mathsf{splitPrint}\})}\mathsf{printed}
 \end{equation*}
 says (using hopefully self-explanatory human-readable syntax for
 disjunctive explicit moves) that `Alice has a strategy to have
 her print job executed, provided that Bob opts to either cancel his
 print job or to split it into smaller
 jobs'. The fixpoint formula
 \begin{equation*}
   \nu x.\,\neg\mathsf{corrupted} \land \clbox{\mathsf{ECC}, (\mathsf{Env}\colon \{0\mathsf{\text{-}flips}, 1\mathsf{\text{-}flip}\})}x
 \end{equation*}
 expresses that ECC memory can ensure that the stored data is not
 corrupted provided that in each cycle the environment flips either
 one or zero bits.
 The formula
 \[\nu x.\,\neg\mathsf{intrusion}\land\cldiamond{\mathsf{Attacker},
     (\mathsf{IPS}\colon \{\mathsf{dropPackage}, \mathsf{blockIP}\})}x\]
 expresses that `No matter what an attacker tries, the intrusion
 prevention system can always drop suspicious packets or block his IP
 address to prevent illegitimate access to company resources'.

\end{example}
\begin{remark}
  One can encode an extension~\acs{ATLDES} of \ac{ATL} with
  disjunctive explicit strategies into the \ac{AMCDES}, e.g.\ defining
  $\ATLdiamond{C,O}(G\,\phi)$ `$C$ can enforce that~$\phi$ always
  holds, provided that $\Ag(O)$ are committed to play strategies
  in~$O$' as
  \begin{equation*}
    \ATLdiamond{C,O}(G\,\phi):= \nu x.\,\phi\land\clbox{C,O} x.
  \end{equation*}
  The~\ac{AMCDES} is more expressive than \acs{ATLDES} in this sense;
  e.g.\ for $C=\{\mathsf{client}\}$ and
  $O=(\mathsf{server}\colon\{\mathsf{protocol},\mathsf{recover}\})$ the
  formula
  $\nu x.\,\mu
  y.\,(\mathsf{granted}\,\land\,\clbox{C,O}x)\,\lor\clbox{C,O}y$ says
  that `$\mathsf{client}$ can enforce that his requests are
  $\mathsf{granted}$ infinitely often, provided that $\mathsf{server}$
  always either keeps to the $\mathsf{protocol}$ or immediately
  $\mathsf{recover}$s when failures occur' (a specification that may,
  of course, hold or fail in a given system).

  Note that the definition of $\ATLdiamond{C,O}$ allows $\Ag(O)$ to
  choose their joint move from~$O$ anew in each step, like in the fixpoint
  formulae of \cref{expl:formulae}, which in fact belong to the
  \acs{ATLDES} fragment of the \ac{AMCDES}. To illustrate that this is
  really the reasonable choice of a semantics for \acs{ATLDES} (as
  opposed to letting~$O$ choose only in the beginning of a play),
  consider a situation where players $K$ (\emph{Kangaroo}) and $M$
  (\emph{Marc-Uwe})~\cite{Kangaroo} play rock-paper-scissors ($R$,
  $P$, $S$) for an indefinite number of rounds, say to determine daily
  who does the dishwashing, until someone quits. Let the model include
  memory for the moves in the previous round, and atoms $p$ `at least
  two rounds have been played' and $k$ `$K$ won the previous
  round'. Consider the \acs{ATLDES} formula
  \begin{equation*}
    \mathsf{rigged} = \ATLdiamond{K,(M\colon\{R,P,S\})} G(p \to k)
  \end{equation*}
  `$K$ wins all rounds after the first if $M$ keeps playing'. In
  \ac{ATLDES}, $\mathsf{rigged}$ does not hold in the model, as one
  would expect. If $M$ could make his choice of $R,P,S$ only once (in
  reality, sadly, he does just that~\cite{Kangaroo2}), then
  $\mathsf{rigged}$ would in fact hold.
\end{remark}
We proceed to define the semantics, which is based on
concurrent game structures~\cite{AlurEA02} extended with
interpretations of explicit moves.
\begin{notn}
  For $k\in\Nat$, we write $\moves{}=\srng{k}$. For
  $C\subseteq \Agents$ and a tuple ${(k_j)}_{j\in C}\in\Nat^{C}$, we
  put $\moves{_C}=\prod_{j\in C}\moves{_j}$. Given $m\in \moves{_C}$
  and $D\subseteq C$, we write $m|_D$ for the restriction of~$m$ to an
  element of $\moves{_D}$. We write \(n \sqsubseteq m\) if
  $n=m|_{\Ag(n)}$, and \(n \ieq m\) if \(n|_{\Ag(n)\cap \Ag(m)} =
  m|_{\Ag(n)\cap \Ag(m)}\). We write $\Pow X$ for the powerset of a set~$X$.
\end{notn}


\begin{defn}\label{def:cgs}
  A \defemph{\acf{CGSES}} 
  is a tuple
  \((W, k, v, f, \interpret)\) consisting of
  \begin{itemize}
    \item a finite set \(W\) of \defemph{states},
    \item for each agent \(j\) and each state \(w\), a natural number
      \(k_j^w\ge 1\) determining the set of \defemph{moves} available to agent~\(j\) at state~\(w\) to
      be~$\moves{_j^w}$, 
    \item for each state \(w \in W\),
      \begin{itemize}
      \item a set \(v(w) \subseteq \At\) of
      propositional atoms true at~\(w\),
      \item an \defemph{outcome function}
        \(f^w : \moves{_\Agents^w} \to W\), and
      \item for each agent $j$, a \defemph{move interpretation}
        \(\interpret^w_j \colon M_j \to \moves{_j^w}\).
      \end{itemize}

  \end{itemize}
\end{defn}
For a joint explicit move~$m\in M_D$, we just write $\interpret^w(m)$
for the joint move with components $\interpret^w_j(m_j)$ for $j\in D$. We use function image notation \(\interpret^w[O]\) to denote the result of applying \(\interpret^w\) to each joint move in the set \(O\). The
semantics of the \ac{AMCDES} is then defined by assigning to each
formula $\phi$ an extension $\Sem{\phi}^\sigma_S\subseteq Q$, which
depends on a \ac{CGSES} \(S = (W, k, v, f, \interpret)\) and a valuation
\(\sigma : V \to \Pow W\).  The propositional cases are standard
(e.g.\ \(\Sem{p}_S^\sigma = \{w\in W\mid p\in v(w)\}\),
\(\Sem{x}_S^\sigma = \sigma(x)\), \(\Sem{\top}_S^\sigma =W\), and
\(\Sem{\phi \wedge \psi}_S^\sigma = \Sem{\phi}_S^\sigma \cap
\Sem{\psi}_S^\sigma\)).  The remaining clauses are
  \begin{align*}
    \Sem{\clbox{C,O}\phi}_S^\sigma &= \{w \in W \mid \exists m_C \in \moves{_C^w}.\, \forall m_\Agents \in \moves{_\Agents^w}.\,\\
                     & \hspace{7em}(m_C \sqsubseteq m_\Agents\land m_\Agents|_{\Ag(O)} \in \interpret^w[O]) \Rightarrow f^w(m_\Agents) \in \Sem{\phi}_S^\sigma\}\\
     \Sem{\cldiamond{C,O}\phi}_S^\sigma &= \{w \in W \mid \forall m_C \in \moves{_C^w}.\, \exists m_\Agents \in \moves{_\Agents^w}.\,\\
                      &\hspace{7em} m_C \sqsubseteq m_\Agents\land m_\Agents|_{\Ag(O)} \in \interpret^w[O] \land f^w(m_\Agents) \in \Sem{\phi}_S^\sigma\}\\
    \Sem{\mu x.\, \phi(x)}_S^\sigma &= \textstyle\bigcap\{B \subseteq W \mid \Sem{\phi(x)}_S^{\sigma[x \mapsto B]} \subseteq B\}\\
    \Sem{\nu x.\, \phi(x)}_S^\sigma &= \textstyle\bigcup\{B \subseteq W \mid B\subseteq\Sem{\phi(x)}_S^{\sigma[x \mapsto B]}  \}
  \end{align*}
  where \(\sigma[x \mapsto B]\) denotes \(\sigma\) updated to
  return~$B$ on input~\(x\);
  and~$\Sem{\cldiamond{C,O}\phi}_S^\sigma=\Sem{\neg\clbox{C,O}\neg\phi}_S^\sigma$.
  That is, $\mu$ and $\nu$ take least and greatest fixpoints according
  to the Knaster-Tarski fixpoint theorem. At a
  state~$w$, $\clbox{C,O}\phi$ holds if the agents in~$C$ have a joint
  move such that a state satisfying~$\phi$ is reached no matter what the
  other agents do, as long as the agents in $\Ag(O)$ play one of the
  joint moves in~$O$. Dually, \(\cldiamond{C,O}\phi\) holds at~$w$
  if whatever the agents in~$C$ do, the other agents have a joint
  move that leads to an outcome in~$\phi$ and in which the joint
  move of $\Ag(O)$ is in~$O$.
  \begin{remark}
    In the modal operators $\clbox{C,O}$, $\Ag(O)$ is in opposition
    to~$C$. One may envision an alternative setup where $\Ag(O)$ is
    instead made a part of~$C$. However, then $\clbox{C,O}\phi$ would
    become equivalent to $\bigvee_{m \in O}\clbox{C,\{m\}}\phi$, hence
    expressible already in \ac{ATLES}. We thus opt for our present
    more expressive version where $\Ag(O)$ and~$C$ are disjoint. Note
    that $\clbox{C,O}\phi$ then is \emph{not} equivalent to
    $\bigwedge_{m \in O}\clbox{C,\{m\}}\phi$: The latter formula
    allows~$C$ to use different moves against each~$m\in O$, while in
    $\clbox{C,O}\phi$, the \emph{same} joint move of~$C$ must work
    against every $m\in O$.
  \end{remark}
  \begin{remark}\label{rem:hist}
    The above semantics uses history-free strategies (i.e.\ ones that
    look only at the present state, not the history of previously
    visited states). While basic \ac{ATL} is insensitive to whether it
    is interpreted over history-free or history-dependent
    strategies~\cite{AlurEA02}, \ac{ATLES} does distinguish these
    semantics~\cite{WaltherEA07}. Although this may not be always
    apparent from the phrasing, all technical results on \ac{ATLES} in
    Walther et al.~\cite{WaltherEA07} are meant to apply to the
    semantics over history-free strategies only\footnote{Personal
      communication with the authors} (in particular the fixpoint
    unfolding axioms~\cite[Figure~1]{WaltherEA07} clearly hold only
    over the history-free semantics). Note that the basic \ac{AMC},
    which the \ac{AMCDES} extends, similarly is interpreted over
    history-free strategies 
    (and nevertheless includes \ac{ATL}$^*$, which is
    history-dependent~\cite{AlurEA02}). 
\end{remark}
\begin{remark}\label{rem:grand}
  The interdiction of proper strategy disjunction in grand coalition
  modalities is needed (only) for the upper bound on satisfiability
  checking (\cref{sec:amcdes-sat}); our results on model checking
  (\cref{sec:amcdes-mc}) would actually not need this restriction.
  The fragment we term \ac{AMCES} in \cref{def:syntax-amcdes} does
  include grand coalition modalities with (non-disjunctive) explicit
  strategies. It is hence more permissive on these modalities than the
  original version of \ac{ATLES}~\cite{WaltherEA07}, where the set of
  agents is made variable, which for purposes of satisfiability is
  equivalent to excluding grand coalition modalities.

  We note that the axiomatization we present later and its
  completeness proof become much simpler if one excludes the grand
  coalition completely (like, effectively, in \ac{ATLES}): E.g.\ in
  the rule $(C)$ for basic coalition logic / \ac{ATL}
  (\cref{sec:coalg-amc}), the literals $\cldiamond{\Agents}c_j$
  disappear; and in the proof of one-step tableau completeness
  (\cref{thm:cl-complete}), one can, in this simplified setting, just
  use a single move $\bot$ as witness for all $\cldiamond{C_j}c_j$
  in~$\Xi$, using non-determinism to ensure satisfaction of the
  $\cldiamond{C_j}c_j$.  This is discussed in detail in
  \cref{sec:remarks-one-step}.
\end{remark}


\paragraph*{Model Checking}\label{sec:amcdes-mc}
Walther et al.~\cite{WaltherEA07} consider two variants of the model
checking problem that differ on whether the interpretation of explicit
strategies is considered part of the model (\emph{fixed}) or to be
found by the model checking algorithm (\emph{open}). They show for \ac{ATLES} that
if strategies are restricted to be history-free, then the problem is
\PT-complete under fixed interpretation, and \NP-complete under open
interpretation, with the upper bound being by straightforward guessing
of history-free strategies. The complexity for the history-dependent
variant remains open.

We obtain upper bounds for model checking in the \ac{AMCDES} using generic
results on the coalgebraic $\mu$-calculus~\cite{DBLP:conf/concur/HausmannS19}:
\begin{theorem}\label{thm:os-mc}
  Model checking for the full \ac{AMCDES} is in \(\NP \cap \coNP\) as
  well as in~$\QP$ under fixed interpretation of explicit strategies,
  and in $\NP$ under open interpretation.
\end{theorem}
We defer a summary of the requisite results in coalgebraic logic and
the proof of \cref{thm:os-mc} to \cref{sec:append-amcd-model}, as the
details are mostly by simple adaptation from the
\ac{AMC}~\cite{DBLP:conf/concur/HausmannS19}.

\section{Preliminaries: Coalgebraic Logic}\label{sec:coalg-log}
We will employ the machinery of coalgebraic logic to obtain our main
complexity results; we recall basic definitions and tools, using the
standard \ac{AMC} as our running example.

Coalgebraic logic~\cite{CirsteaEA11a} is a uniform framework for modal
and temporal logics interpreted over state-based systems. It
parametrizes the \textbf{semantics} of logics over the type of such
systems, encapsulated in a \emph{functor}~$\F$ on the category of
sets. Such a functor assigns to each set~$X$ a set $\F X$ and to each
map $f:X\to Y$ a map $\F f:\F X\to \F Y$, preserving identities and
composition. We think of the elements of~$\F X$ as structured
collections over~$X$. Systems are then \emph{$\F$-coalgebras}, \ie
pairs $(W,\gamma)$ consisting of a set~$W$ of \emph{states} and a
\emph{transition map} $\gamma:W\to \F W$, which thus assigns to each
state a structured collection of successors. 
Our leading example is the functor \(\Gp\) that maps a set~$X$ to the
set
\[\Gp X=\{({(k_j)}_{j\in \Agents},f)\mid (k_j)\in\Nat_{\ge 1}^\Agents,
  f:\textstyle(\prod_{j\in \Agents}\moves{_j})\to X\}\] of
\emph{one-step games} over~$X$. $\Gp$-Coalgebras are essentially
\acp{CGS}~\cite{AlurEA02} without the interpretation of propositional
atoms, as they assign to each state numbers~$k_j$ of available moves
for the agents and an outcome function~$f$. Propositional atoms are
covered by extending~$\Gp$ to $\GP X=\Pow \At \times \Gp X$; although
the logic becomes trivial without propositional atoms, we mostly elide
their explicit treatment, which is straightforward and can be dealt
with using fusion results in coalgebraic
logic~\cite{SchroderPattinson11}. To obtain \acp{CGSES}, we extend
$\Gp$ to the functor $\GES$ with $\GES X$ consisting of \emph{one-step
  games with explicit strategies} $((k_j),f,\interpret)$ over~$X$,
where $((k_j),f)$ is a one-step game over~$X$
and~$\interpret_j\colon M_j\to \moves{_j}$ (for $j\in \Agents$)
interprets explicit strategies; we use the same notation
for~$\interpret$ as introduced for~$\interpret^w$ in
\cref{sec:amcdes}.

The \textbf{syntax} of coalgebraic logics is then parametrized over
the choice of a set~$\Lambda$ of (next-step) \emph{modal operators}
with assigned finite arities; nullary modalities are just
propositional atoms. For readability, we assume in the technical
treatment that all modalities are unary. We require that for every
$\hearts\in\Lambda$ there is a \emph{dual} operator
$\overline\hearts\in\Lambda$. The \emph{coalgebraic
  $\mu$-calculus}~\cite{CirsteaEA11b} over~$\Lambda$ then has formulae
$\phi,\psi$ given by the grammar
\begin{equation*}
  \phi,\psi::=\top\mid\bot\mid x\mid\phi\land\psi\mid\phi\lor\psi\mid\hearts\phi\mid
  \mu x.\,\phi\mid\nu x.\,\phi
\end{equation*}
where $x$ ranges over a reservoir~$V$ of \emph{fixpoint variables},
and~$\hearts$ over~$\Lambda$. The operators~$\mu$ and~$\nu$ take least
and greatest fixpoints, respectively. Again, negation is definable. We
assume a representation of the modalities in~$\Lambda$ as strings over
some alphabet, with an ensuing notion of \emph{representation size}
for formulae and modalities.

Over $F$-coalgebras, a modal operator $\hearts\in\Lambda$ is
interpreted by assigning to it a \emph{predicate lifting}
$\Sem{\hearts}$, which is a family of maps $\Sem{\hearts}_X$, indexed
over all sets~$X$, that assign to each subset $Y\subseteq X$ a subset
$\Sem{\hearts}_X(Y)\subseteq \F X$, subject to a naturality
condition. To enable fixpoint formation, we require $\Sem{\hearts}_X$
to be monotone \wrt subset inclusion. Moreover, we require predicate
liftings to respect duals, \ie
$\Sem{\overline\hearts}_X(Y)=\F X\setminus\Sem{\hearts}_X(X\setminus
Y)$.  Given an $F$-coalgebra $C=(W,\gamma)$ and a valuation
$\sigma:V\to\Pow W$, the semantic clauses defining the extension
$\Sem{\phi}^\sigma_C\subseteq W$ of a formula~$\phi$ are then the
standard ones for the Boolean connectives; $\mu$ and~$\nu$ take least
and greatest fixpoints in the same way as made explicit for the
\ac{AMCDES} in \cref{sec:amcdes}; and
\begin{equation*}
  \Sem{\hearts\phi}_C^\sigma = \gamma^{-1}[\Sem{\hearts}_W(\Sem{\phi}^\sigma_C)].
\end{equation*}
\emph{We fix the data $F$, $\Lambda$, $\Sem{\hearts}$ for
  the remainder of this section.}

\begin{example}\label{ex:liftings}
  The \ac{AMC} is cast as a coalgebraic \MC by interpreting the
  modality~$\clbox{C}$ over the functor~$\Gp$ by the predicate lifting
  \begin{equation*}
    \Sem{\clbox{C}}_X(Y)=\{((k_j),f)\in \Gp X\mid \exists m_C\in \moves{_C}.\,\forall m\in \moves{_\Agents}.\,
    m_C\sqsubseteq m\Rightarrow f(m)\in Y \}
  \end{equation*}
  (using notation introduced in \cref{sec:amcdes}). The more general
  modalities $\clbox{C,O}$ of \ac{AMCDES} are interpreted by a
  predicate lifting that correspondingly lifts a predicate $Y$ on~$X$
  to the set of all one-step games with explicit strategies
  $((k_j),f,\interpret)\in\GES X$ such that there exists a joint move
  $m_C\in \moves{_C}$ such that $f(m)\in Y$ for all
  $m\in \moves{_\Agents}$ such that $m_C\sqsubseteq m$ and
  $\interpret(n) \sqsubseteq m$ for some $n\in O$.
\end{example}
\textbf{Satisfiability checking} in coalgebraic logics can be based
on the provision of a complete set of tableau rules for the next-step
modal operators~\cite{SchroderPattinson09,CirsteaEA11b}. The basic
example of such a rule is the tableau rule
$\mathop{\Box} a_1,\dots,\mathop{\Box} a_n,\mathop{\Diamond}
b/a_1,\dots,a_n,b$ for standard modal logic, which says essentially
that in order to satisfy
$\mathop{\Box} a_1\land\dots\land\mathop{\Box}
a_n\land\mathop{\Diamond} b$, we need to generate a successor state
satisfying $a_1\land\dots\land a_n\land b$. Formal definitions are as
follows.
\begin{defn}[One-step tableau rules]\label{def:rules}
  Fix a supply~$\PV$ of \emph{(propositional) variables}, serving as
  placeholders for formulae in rules. A \emph{(monotone) one-step
    (tableau) rule} has the form
  \begin{equation*}
    \frac{\Phi}{\Theta_1\mid\dots\mid\Theta_n}\qquad(n\ge 0)
  \end{equation*}
  where the \emph{conclusions} $\Theta_1,\dots,\Theta_n$ are finite
  subsets of $\PV$, read as finite conjunctions, and the
  \emph{premiss} $\Phi$ is a finite subset of the set
  $\Lambda(\PV)=\{\hearts a\mid \hearts\in\Lambda,a\in\PV\}$ of
  \emph{modal atoms}, also read conjunctively; additionally, we
  require that~$\Phi$ mentions each variable at most once, and
  the~$\Theta_i$ mention only variables occurring in~$\Phi$. Given a
  set~$X$ and a $\Pow X$-valuation $\tau:\PV\to\Pow X$, we interpret
  such a $\Theta_i$ as
  $\Sem{\Theta_i}\tau=\bigcap_{a\in\Theta_i}\tau(a)$, and $\Phi$ as
  \( \Sem{\Phi}\tau=\textstyle\bigcap_{\hearts
    a\in\Phi}\Sem{\hearts}_X(\tau(a))\subseteq FX.  \)

  The rule $\Phi/\Theta_1\mid\dots\mid\Theta_n$ is \emph{one-step
    tableau sound} if $\Sem{\Theta_i}\tau\neq\emptyset$ for some~$i$
  whenever $\Sem{\Phi}\tau\neq\emptyset$. Let~$\Rules$ be a set of
  one-step tableau rules, closed under injective renaming of
  variables. Then~$\Rules$ is \emph{one-step tableau complete} if the
  following condition holds: For all $X$, $\tau:\PV\to\Pow X$, and
  $\Xi\subseteq\Lambda(\PV)$, whenever for each rule
  $\Phi/\Theta_1\mid\dots\mid\Theta_n\in\Rules$ such that
  $\Phi\subseteq\Xi$, we have $\Sem{\Theta_i}\tau\neq\emptyset$
  for some~$i$, then $\Sem{\Xi}\tau\neq\emptyset$.
\end{defn}
We will give one-step tableau sound and complete sets of
rules for the~\ac{AMCDES} in \cref{sec:amcdes-sat}. To obtain
complexity results, rule sets formally need to be
\emph{$\ET$-tractable}, meaning that rule matches are encodable as
strings over some alphabet such that all rule matches to a given set
of formulae can be represented by polynomially sized codes and
moreover basic operations on codes (well-formedness check, check for
rule matching, access to conclusions) can be performed in exponential
time~\cite{SchroderPattinson09,CirsteaEA11b}; we refrain from
elaborating details, as all rule sets we consider here will be clearly
computationally harmless. The main benefit that we draw from these
rule sets is the following generic upper complexity bound.
\begin{theorem}[Satisfiability checking~\cite{CirsteaEA11b}]\label{thm:coalg-mu}
  If a coalgebraic $\mu$-calculus admits an $\ET$-tractable one-step tableau
  complete set of one-step tableau sound rules, then its
  satisfiability problem is in $\ET$.
\end{theorem}
In the algorithm underlying the above theorem, one-step
rules combine with standard tableau rules for propositional and
fixpoint operators. The arising tableaux need to be checked for bad
branches (where least fixpoints are unfolded indefinitely) using
dedicated parity automata, which combine with the tableau to form the
\emph{tableau game}, a parity game that is won by $\mathsf{Eloise}$
iff the target formula is satisfiable.

\section{Set-Valued First-Order Resolution}\label{sec:set-valued-resolution}
For use in completeness proofs of modal rules, we next introduce
\defemph{set-valued first-order resolution}, an adaptation of the
standard first-order resolution method
\cite{DBLP:books/daglib/0082098} to a logic of \emph{outcome models}
$\omod=((S_j)_{j\in \Agents},f,W,\Sem{-})$ where the~$S_j$ are sets
and~$W$ is a finite set, $\Sem{-}$ interprets sorted algebraic
operations over the~$S_j$, and
$f:\big(\prod_{j\in \Agents}S_j\big)\to W$ is an outcome
function. One-step games in $\Gp W$ are (operation-free reducts of)
outcome models where the~$S_j$ are finite; for the time being, we
allow infinite~$S_j$ for readability, explaining in the proof sketches
in \cref{sec:coalg-amc,sec:amcdes-sat} how finiteness can be
regained. Formulae of \emph{set-valued first-order logic} are clause
sets formed over literals of the form $A(\bar t)$ where $A\subseteq W$
and $\bar t$ is an $\Agents$-tuple of terms (i.e.\ a clause is a
finite set of literals, read disjunctively, and a clause set is a
finite set of clauses, read conjunctively). Terms live in a sorted
setting with one sort~$j$ (interpreted as~$S_j$) for each agent~$j$,
and the $j$-th term in~$\bar t$ has sort~$j$. Terms are built from
sorted variables and function symbols with given sort profiles (e.g.\
$g:1\times 0 \to 2$ takes moves of agents~$1$ and $0$, and produces a
move of agent~$2$) in the standard way, ensuring
well-sortedness. Function symbols are interpreted as sorted functions
on the~$S_j$, respecting the sort profile; this induces an
interpretation of (tuples of) terms depending on sort-respecting
valuations of the variables as usual. We write $\Sem{\bar t}\eta$ for
the interpretation of a tuple~$\bar t$ of terms under a
valuation~$\eta$. An outcome model~$\omod$ as above
\defemph{satisfies} a literal $A(\bar t)$ under a valuation~$\eta$
(notation: $\omod,\eta\models A(\bar t)$) if
$f(\Sem{\bar t}\eta)\in A$, and $\omod$ satisfies a clause~$\Gamma$
under~$\eta$ (notation: $\omod,\eta\models\Gamma$) if
$\omod,\eta\models A(\bar t)$ for some literal $A(\bar t)$
in~$\Gamma$. Finally, $\omod$ \emph{satisfies} a clause~$\Gamma$
(notation: $\omod\models\Gamma$) if $\omod,\eta\models\Gamma$ for every
valuation~$\eta$. A clause set is satisfiable if there exists an
outcome model that satisfies all its clauses. We will generate clauses
from modal atoms in $\Lambda(\PV)$ (\cref{def:rules}); e.g.\ given a
$\Pow W$-valuation $\tau:\PV\to\Pow W$, modalized atoms $\clbox{C}a$
and $\cldiamond{C}a$ induce singleton clauses of the form
\begin{align}
  \label{eq:box-lit} & \{\tau(a)(e_{C},x_{\,\other {C}})\} && \text{(for $\clbox{C}a$)}\\
  \label{eq:diamond-lit} & \{\tau(a)(x_{C},g_{\,\other {C}}(x_{C}))\} && \text{(for $\cldiamond{C}a$)}
\end{align}
respectively, where $x_C$, $x_{\,\other C}$ are tuples of variables
(implicitly universally quantified, and representing moves for the
agents in~$C$ and $\other C$, respectively); $e_C$ is a family of
Skolem constants witnessing the ability of~$C$ to force~$a$; and
$g_{\,\other C}$ is a family of Skolem functions producing
countermoves $g_{\,\other C}(x_C)$ for the agents in $\other C$ that
keep~$C$ from enforcing $\neg a$ using~$x_C$. Of course these symbols
are fresh so that clauses induced by different modalized atoms have
disjoint sets of function symbols and variables, which we will later
distinguish via superscripts in proofs.

We implicitly normalize clauses to mention each tuple of terms at most
once (rewriting $A(\bar t),B(\bar t)$ into $(A\cup B)(\bar t)$), and
operate on clauses using the \defemph{(set-valued) resolution rule}
\begin{equation*}
  (SR)\;\frac{\Gamma,A_1(\bar t_1),\dots,A_n(\bar t_n)\qquad B(\bar u),\Delta}
  {\Gamma\sigma,((\bigcup_{i = 1}^nA_i)\cap B)(\bar u\sigma),\Delta\sigma}\;(n\ge 1)
\end{equation*}
where $\sigma$ is the most general unifier (mgu) of
$\bar t_1, \dots, \bar t_n$, and $\bar u$, with variables in the
premises made disjoint by suitable renaming; as usual, we write `,'
for union of clauses and omit set brackets around singleton clauses
(so $\Gamma,A(\bar t)$ is shorthand for $\Gamma\cup\{A(\bar t)\}$). We
will also need to consider a liberalized variant~$(\mathit{lSR})$
of~$(\mathit{SR})$ where we only require~$\sigma$ to be a unifier
(rather than the mgu). Note that the above formulation of $(SR)$
includes \emph{one-sided implicit factoring}, i.e.\ several
literals in the left-hand clause may be resolved against one literal
in the right-hand clause. A clause is \defemph{blatantly inconsistent}
if all its literals are of the form $\emptyset(\bar t)$. A clause
set~$\phi$ is \defemph{blatantly inconsistent} if it contains a
blatantly inconsistent clause, and \defemph{inconsistent} if a
blatantly inconsistent clause can be derived from it using the
resolution rule; otherwise,~$\phi$ is \defemph{consistent}. We note
that the liberalized rule~$(\mathit{lSR})$ is admissible; that is,
calling a clause set \emph{$(lSR)$-consistent} if no blatantly
inconsistent clause can be derived from it using $(\mathit{lSR})$, we
have
\begin{lemma}\label{lem:lsr-admissible}
  Consistent clause sets are also $(\mathit{lSR})$-consistent.
\end{lemma}
\begin{proof}
  By straightforward induction on derivations under~$(\mathit{lSR})$,
  one shows that every clause derivable from a clause set~$\phi$ under
  $(\mathit{lSR})$ is a substitution instance of a clause derivable
  from~$\phi$ under $(\mathit{SR})$. The claim follows immediately by
  the observation that a clause that has a blatantly inconsistent
  substitution instance must itself be blatantly inconsistent.
\end{proof}

\label{page:occurs} Recall that unification can fail either due to a
\emph{clash}, i.e.\ when terms with distinct head symbols need to be
unified, or at the \emph{occurs check}, which happens when a variable
needs to be unified with a term that contains it. In particular, this
happens in clauses~\eqref{eq:diamond-lit} associated with diamonds:
E.g.\ the modal atoms $\cldiamond{\{0\}}a$ and $\cldiamond{\{1\}}b$
generate clauses $\{\tau(a)(x_0,g^1_1(x_0))\}$ and
$\{\tau(b)(g^2_0(x_1),x_1)\}$, whose (tuples of) argument terms fail
to unify since no substitution solves $x_0=g_0^2(g_1^1(x_0))$.

We note that from one-sided implicit factoring as incorporated in~$(SR)$,
we can derive two-sided implicit factoring:
\begin{lemma}\label{lem:two-sided-derivable}
  From~$(lSR)$, the rule
\begin{equation*}
  (\mathit{SR}^+)\;\frac{\Gamma,A_1(\bar t_1),\dots,A_n(\bar t_n)\qquad B_1(\bar u_1),\dots,B_m(\bar u_m),\Delta}
  {\Gamma\sigma,((\bigcup_{i = 1}^nA_i)\cap (\bigcup_{i = 1}^mB_i))(\bar u_1\sigma),\Delta\sigma}\;(n,m\ge 1)
\end{equation*}
is derivable where $\sigma$ is the mgu of
$\bar t_1, \dots, \bar t_n,\bar u_1.\dots,\bar u_m$, with variables in
the premises made disjoint by suitable renaming.
\end{lemma}
\begin{proof}
  Using $(\mathit{lSR})$, we derive from the premises the clause
  $\Gamma\sigma,((\bigcup_{i = 1}^nA_i)\cap B_1)(\bar
  u_1\sigma),B_2(\bar u_2),\dots,B_{m}(\bar
  u_{m}\sigma),\Delta\sigma$. Further resolving this clause with the
  left premise according to $(\mathit{lSR})$, we obtain the clause
  $\Gamma\sigma,((\bigcup_{i = 1}^nA_i)\cap B_1)(\bar
  u_1\sigma),((\bigcup_{i = 1}^nA_i)\cap B_2)(\bar
  u_2\sigma),\dots,B_{m}(\bar u_{m}\sigma),\Delta\sigma$. Since
  $\bar u_1\sigma=\bar u_2\sigma$, this clause is identified with
  $\Gamma\sigma,((\bigcup_{i = 1}^nA_i)\cap (B_1\cup B_2))(\bar
  u_1\sigma),\dots,B_{m}(\bar u_{m}\sigma),\Delta\sigma$ (note that
  $((\bigcup_{i = 1}^nA_i)\cap B_1)\cup ((\bigcup_{i = 1}^nA_i)\cap
  B_2)= (\bigcup_{i = 1}^nA_i)\cap (B_1\cup B_2)$). Continuing in this
  manner, we obtain the conclusion of $(\mathit{SR}^+)$.
\end{proof}

\emph{Set-valued propositional resolution} in \emph{set-valued
  propositional logic} simplifies the above setup by replacing
tuples~$\bar t$ of terms in literals $A(\bar t)$ with elements~$y$ of
some index set~$Y$; models are then just functions $f\colon Y\to W$,
and~$f$ \emph{satisfies} a literal $A(y)$ if $f(y)\in A$. The
resolution rule is just like the above but of course does not involve
unification, substitution, and implicit factoring, i.e.\ it just
derives $\Gamma,(A\cap B)(y),\Delta$ from $\Gamma,A(y)$ and
$B(y),\Delta$.
\begin{theorem}[Soundness and completeness of set-valued
  resolution]\label{thm:res-completeness}
  A clause set in set-valued propositional (first-order) logic is
  satisfiable iff it is consistent under set-valued propositional
  (first-order) resolution.
\end{theorem}
\begin{proof}[Proof sketch]
  Soundness (`only if') is clear (see \cref{sec:proof-thm:r-comp}). Completeness (`if') of the
  propositional variant depends on~$W$ being finite. It proceeds via
  maximally consistent clause sets (MCS) and a Hintikka lemma stating
  in particular that an MCS containing $(A\cup B)(y)$ must also
  contain one of $A(y),B(y)$. Completeness of the first-order variant
  is by adaptation of the completeness proof for standard first-order
  resolution, going via Herbrand models (i.e.\ models having the set
  of ground terms as the carrier set) and reduction to completeness of
  set-valued propositional resolution.
\end{proof}
Of course, the Herbrand models constructed in the proof of
\cref{thm:res-completeness} are in general infinite. For purposes of
constructing finite models, we identify a property of `sufficient
completeness' of a model for a set of terms.
\begin{defn}
  A set~$\CT$ of (tuples of) terms is \emph{closed under unification}
  if whenever $t,s\in\CT$ are unifiable and~$\sigma$ is an mgu
  of~$t,s$, then $u\sigma\in \CT$ for every~$u\in\CT$.
\end{defn}
\begin{remark}
  If~$\CT$ is closed under unification, then~$\CT$ is in particular
  closed under injective renaming of variables: For $u\in\CT$, every
  injective renaming~$\sigma$ is an mgu of $u,u$, so that
  $u\sigma\in\CT$.
\end{remark}
We will treat tuples of terms like terms in the following, in
particular mentioning equations between tuples of terms and unifiers
of such equations; this is to be understood as referring to
componentwise equality.

\begin{defn}
  A \emph{solution} of an equation $t=s$ in an outcome model~$\omod$
  is a valuation~$\eta$ such that $\Sem{\bar t}\eta=\Sem{\bar s}\eta$
  in~$\omod$. Let $\CT$ be a set of tuples of terms. We say
  that~$\omod$ is \emph{$\CT$-equationally complete} if whenever an
  equation $\bar t=\bar s$ with $\bar t,\bar s\in\CT$ has a solution
  in~$\omod$, then $\bar t,\bar s$ are unifiable, and the mgu~$\sigma$
  of~$\bar t,\bar s$ is a \emph{most general solution} of
  $\bar t=\bar s$ in~$\omod$, i.e.\ every solution~$\eta$ of
  $\bar t=\bar s$ in~$\omod$ has the form $\eta(x)=\Sem{\sigma(x)}\eta'$
  for some valuation~$\eta'$; we then say briefly that~$\eta$
  \emph{factorizes through}~$\sigma$.
\end{defn}

\begin{theorem}\label{thm:cmpl-model-satisf}
  Let $\CT$ be a set of tuples of terms that is closed under
  unification, and let~$\omod$ be $\CT$-equationally
  complete. Let~$\phi$ be a clause set such that $\bar t\in\CT$ for
  every literal $B(\bar t)$ occurring in~$\phi$. If~$\phi$ is
  consistent under set-valued first-order resolution, then~$\phi$
  is satisfiable over~$\omod$.
\end{theorem}
\begin{proof}
  By completeness of set-valued propositional resolution
  (\cref{thm:res-completeness}), it suffices to show that the clause
  set~$\phi^\omod$ consisting of all instances over~$\omod$ of clauses
  in~$\phi$ is consistent under set-valued propositional
  resolution. Formally, an instance $\Sem{\Gamma}\eta$ over~$\omod$ of
  a clause~$\Gamma$ is induced by an $A$-valuation~$\eta$, and given
  as
  \begin{equation*}
    \Sem{\Gamma}\eta=\{B(\Sem{\bar t\,}\eta)\mid B(\bar t\,)\in\Gamma\}.
  \end{equation*}
  Since~$\CT$ is closed under unification, we can assume w.l.o.g.\
  (using \cref{lem:lsr-admissible,lem:two-sided-derivable})
  that~$\phi$ is closed under set-valued first-order resolution with
  implicit two-sided factoring, i.e.\ under the rule~$(\mathit{SR}^+)$
  (since all terms that appear when closing~$\phi$ under resolution
  remain in~$\CT$); then it suffices to show that~$\phi^\omod$ is
  closed under set-valued propositional resolution, since~$\phi$ and,
  hence,~$\phi^\omod$ do not contain blatantly inconsistent clauses.

  So let $\Gamma,A_1(\bar t_1),\dots,A_n(\bar t_n)$ and
  $B_1(\bar s_1),\dots,B_m(\bar s_m),\Delta$ be clauses in $\phi$,
  with variables made disjoint. By the latter restriction, resolvable
  instances of these clauses in~$\omod$ can be assumed to use the same
  valuation; so let~$\eta$ be a valuation such that
  $\Sem{\bar t_1}\eta = \dots = \Sem{\bar t_n}\eta = \Sem{\bar
    s_1}\eta=\dots= \Sem{\bar s_m}\eta$. Then in particular
  $\bar t_1 = \dots = \bar t_n = \bar s_1=\dots=\bar s_m$ is solvable
  in~$\omod$. Since
  $\bar t_1,\dots,\bar t_n,\bar s_1,\dots,\bar s_m \in\CT$, it follows
  by $\CT$-equational completeness of~$\omod$ that
  $\bar t_1,\dots,\bar t_n,\bar s_1,\dots,\bar s_m$ are unifiable,
  hence have an mgu~$\sigma$, and that $\sigma$ is a most general
  solution of $\bar t_1 = \dots = \bar t_n = \bar s_1=\dots=\bar s_m$
  in~$\omod$. This implies that $\eta$ has the form
  $\eta(x)=\Sem{\sigma(x)}\eta'$ for some $A$-valuation~$\eta'$. Thus,
  the resolvent
  $\Sem{\Gamma,((\bigcup_{i = 1}^nA_i)\cap (\bigcup_{i=1}^mB_i))(\bar
    u),\Delta}\eta$ of the two instances has the form
  $\Sem{\Gamma\sigma,((\bigcup_{i = 1}^nA_i)\cap
    (\bigcup_{i=1}^mB_i))(\bar u\sigma),\Delta\sigma}\eta'$, and hence
  is in~$\phi^\omod$ as required since
  $\Gamma\sigma,((\bigcup_{i = 1}^nA_i)\cap (\bigcup_{i=1}^mB_i))(\bar
  u\sigma),\Delta\sigma$ is in~$\phi$ by closure of~$\phi$ under
  $(\mathit{SR}^+)$. 
\end{proof}


\section{The \ac{AMC},  Coalgebraically}\label{sec:coalg-amc}

To illustrate the use of one-step tableau rules, we briefly indicate how
to obtain the \ET upper bound for the \ac{AMC} by \cref{thm:coalg-mu}.
%
%
%
The requisite functor~$\Gp$ and the associated predicate liftings have
been recalled in \cref{sec:coalg-log}.
%
%
We recall the known rule set~\cite{SchroderPattinson09,CirsteaEA11b}:
\begin{gather*}
  (CD)\;\frac{\clbox{D_1}a_1,\dots,\clbox{D_\alpha}a_\alpha}{a_1,\dots,a_\alpha}
  \\
  (C)\;\frac{\clbox{D_1}a_1,\dots,\clbox{D_\alpha}a_\alpha, \cldiamond{E}b,\cldiamond{\Agents}c_1,\dots,\cldiamond{\Agents}c_\beta}{a_1,\dots,a_\alpha, b,c_1,\dots,c_\beta}
\end{gather*}
where for each \(j,k\), \(D_j \cap D_k = \emptyset\) and
\(D_j \subseteq E\).  Soundness of these rules is straightforward
(they say in particular that disjoint coalitions can combine their
abilities and that coalitions inherit the abilities of subcoalitions);
for illustration, we show one-step tableau completeness using
set-valued resolution (\cref{sec:set-valued-resolution}), alternative
to proofs in the
literature~\cite{Drimmelen03,GorankoDrimmelen06,schewe_PhD}.

\begin{theorem}[One-step tableau completeness]\label{thm:cl-complete}
  The rules $(C)$, $(CD)$ are one-step tableau complete \wrt \ac{AMC}.
\end{theorem}
By \cref{thm:coalg-mu}, this implies the known (tight) $\ET$ upper
bound for satisfiability checking in the~\ac{AMC}~\cite{schewe_PhD}.

\begin{proof}
  As indicated above, we present a proof producing infinite sets of
  moves in one-step games, and then discuss how finiteness of move
  sets is regained using the notion of $\CT$-equationally complete
  (finite) model (\cref{thm:cmpl-model-satisf}).

  Let $\tau$ be a $\Pow W$-valuation, and let
  \(\Xi=\{\clbox{D_1}a_1,\dots,\allowbreak\clbox{D_\alpha}a_\alpha,\allowbreak\cldiamond{C_1}c_1,\dots,\allowbreak\cldiamond{C_\beta}c_\beta\}\)
  such that for every instance of $(C)$ or $(CD)$ that applies to
  (some subset of)~$\Xi$, the conclusion~$\Theta$ satisfies
  $\Sem{\Theta}\tau\neq\emptyset$. We have to show that
  $\Sem{\Xi}\tau\neq\emptyset$. To this end, we translate~$\Xi$ into a
  clause set~$\phi$ in set-valued first-order logic
  (\cref{sec:set-valued-resolution}), generating one (singleton)
  clause for each modalized atom $\clbox{D_j}a_j$ and
  $\cldiamond{C_j}c_j$ according to~\eqref{eq:box-lit}
  and~\eqref{eq:diamond-lit} (\cref{sec:set-valued-resolution}), with
  distinct Skolem constants~$e^j_{D_j}$ and Skolem functions
  $g^j_{\other{C_j}}$, respectively. By \cref{thm:res-completeness},
  it suffices to show that $\phi$ is consistent under set-valued
  resolution.  We observe the following.
  \begin{enumerate}
  \item Two clauses~$b_j$ and $b_k$ of shape~\eqref{eq:box-lit}, for
    $j \neq k$, resolve only if $D_j\cap D_k=\emptyset$ -- otherwise,
    unification fails due to a clash between \(e^j_i\) and \(e^k_i\)
    for each agent $i\in D_j\cap D_k$.
  \item Similarly, a clause~$b_j$ of shape~\eqref{eq:box-lit} resolves
    with a clause~$d_k$ of shape~\eqref{eq:diamond-lit} only if
    $D_j\cap\other C_k=\emptyset$, i.e.\ $D_j\subseteq C_k$.
  \item Similarly, two clauses~$d_j$ and $d_k$ of
    shape~\eqref{eq:diamond-lit}, for $k\neq j$, resolve only if
    $\other{C_j}\cap\other{C_k}=\emptyset$, i.e.\
    $C_j\cup C_k=\Agents$
  \item \label{item:dia-dia} 
    Crucially, two clauses~$d_j$ and~$d_k$ of
    shape~\eqref{eq:diamond-lit} , for $k\neq j$, resolve only if at
    least one of~$C_j$ and~$C_k$ is~$\Agents$: Assume that
    $p\in\other{C_j}$ and $q\in\other{C_k}$. By the previous
    item, $p\in C_k$ and $q\in C_j$, so~$x_p$ is an argument
    in~$g^k_q$ and~$x'_q$ (renamed for purposes of the resolution
    step) is an argument in~$g^j_p$, implying that unification
    of~$d_j$ and~$d_k$ fails at the occurs check (cf.\
    p.~\ref{page:occurs}).  This explains why only one $\cldiamond{E}$
    with $E\neq N$ is needed in rule $(C)$.
  \end{enumerate}
  These observations imply that a resolution proof of a blatantly
  inconsistent (necessarily singleton) clause from~$\phi$ will witness
  a rule match of either~\((C)\) or \((CD)\) (depending on whether
  clauses of shape~\eqref{eq:diamond-lit} are involved), and blatant
  inconsistency means that $\Sem{\Theta}\tau=\emptyset$ for the
  corresponding rule conclusion~$\Theta$, contradicting the assumption
  on~$\Xi$.

  \subparagraph*{Finitely many moves} As indicated in
  \cref{sec:set-valued-resolution}, the model of~$\Xi$ thus produced
  will have infinitely many moves per agent, namely the ground terms
  generated by the Skolem constants and functions. We can replace
  these with finitely many moves where agents play \emph{Skolem
    symbols} paired with \emph{colours} -- simulating the effect of
  the occurs check from the unification procedure -- taken from a
  finite abelian group~$U$ (with neutral element~$0$ and group
  operation~$+$) that contains distinct elements $u_1,\dots,u_\beta$
  (e.g.\ $U=\Int/\beta\Int$).  Specifically, all agents receive (for
  simplicity) the same moves, namely
  \begin{itemize}
  \item moves $(e^j,0)$ for $j=1,\dots,\alpha$, intended as witnesses
    for $\clbox{D_j}a_j$, and
  \item moves $(g^j,u)$ for $j=1,\dots,\beta$ and $u\in U$, intended
    as witnesses for $\cldiamond{C_j}c_j$.
  \end{itemize}
  We refer to the first component of a move as its \emph{move symbol},
  and to the second as its \emph{colour}. By $\col(m_C)$ we denote the
  sum of all colours of the moves in a joint move $m_C$ for~$C$.
  
  Let \(\CT\) be the unification closure of the set of all tuples of
  argument terms occuring in clauses from
  \(\phi\). 
  By the above analysis, all tuples in \(\CT\) essentially have the
  shape \((x_A, e_B, g_{\,\other{A \cup B}}(x_A,e_B))\) where \(x_A\)
  are variables, \(e_B\) are Skolem constants possibly from different
  box modalities, and \(g_{\,\other{A \cup B}}\) are Skolem functions
  from a single diamond (as Skolem functions for different diamonds do
  not initially occur in the same tuple of terms and such occurrences
  are not introduced during unification due to the occurs check); any
  one of $x_A$, $e_B$, $g$ may be absent.  The (finite) model \(\omod\)
  is then defined over coloured moves.  Skolem constants $e^j$ are
  interpreted as $(e^j,0)$, and Skolem functions \(g_i^j\) for
  $i\in \other{C_j}$ are interpreted as mapping a joint move
  $m_{C_j}$ of~$C_j$ to $(g^j,u_j-\col(m_{C_j}))$ if $i$ is the least
  element of~$\other{C_j}$, and to~$(g^j,0)$ otherwise, thus
  ensuring that $\col(m_{C_j},g^j(m_{C_j}))=u_j$. We proceed to show
  that~$\omod$ is $\CT$-equationally complete, obtaining by
  \cref{thm:cmpl-model-satisf} and consistency of~$\phi$ under
  set-valued first-order resolution that~$\phi$ is satisfiable
  over~$\omod$.
   
  So let \(t, u \in \CT\) such that \(t = u\) has a solution \(\eta\)
  in \(\omod\).  We proceed by case distinction on the shape of
  \(t = u\): 

  \((x_A, e_B) = (x'_{A'}, e'_{B'})\): In the simplest case the terms
  just consist of variables ($x_A$, $x'_{A'}$) and Skolem constants
  ($e_B$, $e'_{B'}$).  Given the interpretation of the Skolem
  constants in~$\omod$, it is clear that $e_B$ and $e_B'$ must agree on
  \(B \cap B'\) so $t,u$ are unifiable.  The solution \(\eta\)
  necessarily replaces variables in \(A \cap B'\) and \(A' \cap B\)
  with the respective interpretations of Skolem constants on the other
  side of the equality.  Hence, the solution \(\eta\) factorizes
  through the mgu of \(t\) and~\(u\).

  \((x_{A}, e_{B}, g^j_{\,\other{A \cup B}}(x'_{A},e'_{B})) =
  (x_{A'}, e_{B'})\): This case is similar to the previous one, using
  the observation that given the interpretation of~$g^j$ in~$\omod$, the
  equation can only have a solution if
  $(\other{A \cup B})\cap B'=\emptyset$, i.e.\
  $(\other{A \cup B})\subseteq A'$.

  \((x_A, e_B, g^j_{\,\other{A \cup B}}(x_A,e_B)) = (x'_{A'},
  e'_{B'}, g^k_{\,\other{A' \cup B'}}(x'_{A'},e'_{B'}))\): The
  interpretations of the terms $g^j_{\,\other{A \cup B}}(x_A,e_B)$
  and $g^k_{\,\other{A' \cup B'}}(x_{A'},e'_{B'}))$ in \(\omod\)
  (under~$\eta$) have the form \((g^j, c)\) and \((g^k, d)\) for some
  \(c\) and \(d\), respectively.  The case where $j=k$ is essentially
  like the previous cases. The interesting case is where $j\neq k$, in
  which case necessarily \(\other{A \cup B} \subseteq A'\) and
  \(\other{A' \cup B'} \subseteq A\); this is the case where
  unification of $t,u$ fails at the occurs check as explained above.
  However, the construction of~$\omod$ ensures that now $t=u$ also has
  no solution in~$\omod$, as the respective interpretations of~$g^j$
  and~$g^k$ ensure that the colour of the whole joint move is \(u_j\)
  on the left and \(u_k\) on the right.
\end{proof}
The proof for the \ac{AMCDES} proceeds in a quite similar
fashion, and will be presented in less detail. 

\section{\ac{AMCDES} Satisfiability}\label{sec:amcdes-sat}

We now extend this treatment to obtain \ET satisfiability checking for
\ac{AMCDES}, cast coalgebraically using the functor and predicate
liftings presented in \cref{sec:coalg-log}. We have one-step rules
$\rs$, $\rc$, where $\rc$ is

\begin{equation*}
  \rc\;\frac{
      \clbox{D_1, P_{G_1}}a_1, \dots, \clbox{D_\alpha, P_{G_\alpha}}a_\alpha, \cldiamond{E, Q_K}b,\cldiamond{C_1, r_{H_1}}c_1, \dots, \cldiamond{C_\beta, r_{H_\beta}}c_\beta
  }{{(a_j)}_{j \in \inda_q}, b, {(c_j)}_{j \in \indb_q} \mid \cdots \text{ for } q \in Q_K}
\end{equation*}
(i.e.\ the rule has one conclusion for each~$q$) where $\Ag(Q_K)=K$;
the $r_{H_j}$ are (non-disjunctive) explicit joint moves for
coalitions~$H_j$; $\inda_q \subseteq \srng{\alpha}$,
\(\indb_q \subseteq \srng{\beta}\) for each~$q \in Q_K$; and the
following side conditions hold, with
\(L: = \bigcup^\alpha_{j = 1}G_j \cup \bigcup^\beta_{j =
  1}H_j\): 

\begin{enumerate}
\item\label{item:D_disj} For each \(j,k\), \(D_j \cap D_k = \emptyset\). 
\item\label{item:full_dia} For each \(j\), \(C_j \cup H_j = \Agents\).  
\item\label{item:D_anon} 
  \(\bigcup^\alpha_{j = 1} D_j \cap L = \emptyset\).  
\item\label{item:D_E} 
  \(\bigcup^\alpha_{j = 1}D_j \subseteq E\). 

\item\label{item:no_names_dropped} \(E \cup K \supseteq L\). 
\item\label{x1} \(r_{H_j} \ieq q\) for all
  \(q \in Q_K\), \(j \in \indb_q\). 
\item\label{x2} There is a joint explicit move \(l\) for 
  \(E \cap L\) such that 
  \(r_{H_j} \ieq l\) for each
  \(q \in Q_K\), \(j \in \indb_q\), and moreover for each
  \(j \in \inda_q\) there exists \(p \in P_{G_j}\) such that
  \(p \ieq q\) and
  \(p \ieq l\).  
\end{enumerate}

\noindent Rule $\rs$ is a variant of~$\rc$ obtained by instantiating
to $\cldiamond{E, Q_K}b=\cldiamond{\Agents,\{()\}}\top$,
\(\inda_{()}=\srng{\alpha}\), and
\(\indb_{()}=\srng{\beta}\), and then omitting the (valid)
literal \(\cldiamond{\Agents,\{()\}}\top\) from the rule premiss; side
conditions~\ref{item:D_E}.--\ref{x1}.\ then become
trivial and can be omitted.

Rule $\rc$ extends the rules for the basic AMC as recalled in
\cref{sec:coalg-amc}. The new features are intuitively understood as
follows. Imagine that $D_1,\dots,D_n$ play moves witnessing their
ability to (conditionally) enforce $a_1,\dots,a_n$. According to
$\cldiamond{E,Q_K}$, $K$ can then play some move $q\in Q_K$
additionally ensuring~$b$; the $q$-th conclusion of $\rc$ captures the
constraints on the next state reached in this situation. These
additionally depend on the moves chosen by the remaining agents (those
in~$E\setminus\bigcup D_i$): If the arising joint move restricts to
one of the moves in~$P_{G_j}$, then~$D_j$ successfully enforces~$a_j$,
and if it restricts to~$r_{H_j}$, then the next state must
satisfy~$c_j$ (note that since $C_j\cup H_j=\Agents$,
$\cldiamond{C_i,r_{H_k}} c_j$ says that~$c_j$ is enforced as soon as
$H_j$ play $r_{H_j}$). The index sets~$\inda_q$ and~$\indb_q$ indicate for
which~$j$ this applies, and side conditions~\ref{x1} and~\ref{x2}
ensure that a corresponding joint move actually exists. For
definiteness, we note

\begin{lemma}[One-step soundness]\label{thm:soundness-ess}
  The rules \rs, \rc are one-step tableau sound \wrt \ac{AMCDES}.
\end{lemma}
\begin{proof}
  By the above, it suffices to show soundness of \rc, formalizing the
  above intuitive explanation. Write \(\phi\) for the premiss of the
  rule, and \(\psi_q\) for the conclusion associated to $q\in Q_K$.
  Let $\tau$ be a $\Pow W$-valuation such that
  \(\Sem{\phi}\tau \neq \emptyset\), and fix
  \(\omod = ((k_j), f, \interpret) \in \Sem{\phi}\tau\); we have to
  show that \(\Sem{\psi_q}\tau \neq \emptyset\) for some
  \(q \in Q_K\). We refer to side conditions by their numbers:

    \begin{itemize}
    \item For each \(j \in \srng{\alpha}\), we have a joint move \(e_j\) for \(D_j\) witnessing \(\clbox{D_j,P_{G_j}}a_j\).
      By~\ref{item:D_disj}., the~$e_j$ can be combined into a joint move \(e\) for \(\bigcup_{j = 1}^\alpha D_j\).
      \item By~\ref{item:D_anon}., \(e\) can be combined with (the interpretation of) the explicit move \(l\) postulated in~\ref{x2}.\ into a move \(x_0\) for \((E \cap L) \cup \bigcup_{j = 1}^\alpha D_j\subseteq E\), where the inclusion is by~\ref{item:D_E}.
        Extend \(x_0\) arbitrarily to a move~\(x\) for the whole coalition \(E\).
      \item Since \(\omod \in \Sem{\cldiamond{E,Q_K}b}\) and $\Ag(x) = E$, there is some $q\in Q_K$ and a joint move $m_q$ for~$\Agents$ such that $x,q\sqsubseteq m_q$ and \(f(m_q) \in \tau(b)\).
      \item To obtain that \(f(m_q) \in \Sem{\psi_q}\tau\) for
        this~$q$, it remains to show that $f(m_q)$ satisfies the
        remaining literals~$a_j,c_j$ of~$\psi_q$:
        \begin{itemize}
        \item For $j\in \inda_{q}$, we have
          $e_j\sqsubseteq m_q$ and, by~\ref{item:no_names_dropped}. and~\ref{x2}.,
          \(\interpret[p] \sqsubseteq m_q\) for some \(p \in P_{G_j}\), so that
          $\omod\in\Sem{\clbox{D_j,P_{G_j}}a_j}\tau$ implies
          $f(m_q)\in\tau(a_j)$.
        \item For \(j \in \indb_{q}\), we have
          \(\interpret[r_{H_j}] \sqsubseteq m_q\)
          by~\ref{item:no_names_dropped}.,~\ref{x1}.,
          and~\ref{x2}. Since \(C_j\cup H_j = \Agents\), we thus have
          that \(\omod \in \Sem{\cldiamond{C_k,r_{H_k}}c_k}\tau\)
          implies \(f(m_j) \in \tau(c_k)\). \qedhere
        \end{itemize}
    \end{itemize}
  \end{proof}
It remains to prove ompleteness:

\begin{lemma}[One-step tableau completeness]\label{thm:cldes-complete}
  The rules \rs, \rc are one-step tableau complete \wrt \ac{AMCDES}.
\end{lemma}
\begin{proof}
  Let $\tau$ be a $\Pow W$-valuation, and let
  \(\Xi=\{\clbox{D_1,P_{G_1}}a_1,\dots,\allowbreak\clbox{D_\alpha,P_{G_\alpha}}a_\alpha,\allowbreak\cldiamond{C_1,R_{H_1}}c_1,\dots,\allowbreak\cldiamond{C_\beta,
    R_{H_\beta}}c_\beta\}\) such that every instance of \rs or \rc
  whose premise is contained in~$\Xi$ has a conclusion that is
  non-empty under~$\tau$.  We have to show that
  $\Sem{\Xi}\tau\neq\emptyset$. We translate~$\Xi$ into a clause
  set~\(\phi\) in set-valued first-order logic by including for each
  $\clbox{D_j,P_{G_j}}a_j$ and each $p\in P_{G_j}$ a singleton clause
  \begin{equation}
    \label{eq:boxc}  \{\tau(a_j)(e^j_{D_j},x_{\,\other{D_j \cup G_j}},p)\},
  \end{equation}
  (so $e^j_{D_j}$ witnesses $\clbox{D_j,P_{G_j}}a_j$), and for each
  $\cldiamond{C_j, R_{H_j}}c_j$ a clause
  \begin{equation}
   \label{eq:diac}  \{\tau(c_j)(x_{C_j},g^j_{\,\other{C_j \cup H_j}}(x_{C_j}),r) \mid r \in R_{H_j}\}
  \end{equation}
  (so the $g^j_{\,\other{C_j \cup H_j}}$ are Skolem functions
  witnessing $\cldiamond{C_j, R_{H_j}}c_j$). We now proceed as in the
  proof of \cref{thm:cl-complete}: We first show that \(\phi\) is
  consistent under set-valued resolution, obtaining by
  \cref{thm:res-completeness} that~$\phi$ is satisfiable in a model
  that may have infinitely many moves, and then present a finite
  $\CT$-equationally complete model for the unification closure~$\CT$
  of the involved terms. 
  Write $b^p_j$ for clauses of type~\eqref{eq:boxc}
  for given $j=1,\dots,\alpha$ and $p\in G_j$, and $d_j$ for the
  $j$-th clause of type~\eqref{eq:diac}.

  Unlike in the proof of \cref{thm:cl-complete}, we thus may have
  non-singleton clauses, of shape~\eqref{eq:diac}.  We first note that
  implicit factoring plays no role in resolution from this clause set:
  The non-singleton clause resulting from a diamond
  \(\cldiamond{C_j,R_{H_j}}c_j\) has a unique Skolem constant
  \(r \in R_{H_j}\) in each literal, so its literals do not unify
  among each other.  As unification does not get rid of these
  constants, this restriction will be an invariant throughout
  resolution over this clause set.  However, we shall see that these
  non-singleton clauses do not resolve among each other. We note the
  following observations.
  \begin{enumerate}
  \item $b^p_j$ and $b^q_j$, for $p\neq q$, do not resolve (and
    resolving $b^p_j$ with itself is pointless).
  \item $b^p_j$ and $b^q_k$, for $k\neq j$, resolve only if
    $D_j\cap D_k = D_j \cap G_k = D_k \cap G_j = \emptyset$, and moreover $p \ieq q$. 
  \item $b^p_j$ and $d_k$ resolve, at the $d_k$-literal
    for~$r\in R_{H_k}$, only if $D_j\subseteq C_k$, 
    and hence in particular also \(D_j \cap H_k =
    \emptyset\), 
    $D_j\cup G_j\subseteq C_k\cup H_k$ (equivalently
    $\other{C_k\cup H_k}\subseteq\other{D_j\cup G_j}$), and
    $r\ieq p$.
  \item $d_j$ and $d_k$, for $k\neq j$, resolve, at the $d_j$-literal
    for $r\in H_j$ and the $d_k$-literal for $r'\in H_k$, only if
    $C_j\cup C_k\cup H_k=\Agents$ (equivalently
    $\other{C_k\cup H_k}\subseteq C_j$), $C_k\cup C_j\cup H_j=\Agents$,
    and $r\ieq r'$. 
  \item Like in the proof of \cref{thm:cl-complete}, it follows that
    $d_j$ and $d_k$ resolve only if at least one of
    $\cldiamond{C_j,R_{H_j}}$ and $\cldiamond{C_k,R_{H_k}}$ is a grand
    coalition modality (since otherwise unification fails at the
    occurs check), in which case the corresponding clause is a
    singleton.
  \item Clauses obtained from clauses of shape~\eqref{eq:diac} by
    resolving with singleton clauses retain essentially
    shape~\eqref{eq:diac}, only with some of the variables~$x_i$
    replaced with constants. Resolution of such clauses is thus
    subject to the same restrictions; in particular, non-singleton
    clause of this kind they will not resolve among each other.
  \end{enumerate}
  Thus, a proof of a blatantly inconsistent clause from~$\phi$ by
  set-valued resolution will involve either zero or one clauses
  $d_j$ where $C_j\cup H_j\neq \Agents$.  We will refer to resolution
  proofs of the first kind as type-0 and to proofs of the second kind
  as type-1.

  \subparagraph*{Type-0 proofs} We show that in this case, the
  impossibility of deriving a blatantly inconsistent clause is
  obtained via rule $\rs$.  To apply \rs to the set of modal atoms
  involved in the proof, we need to show the side conditions of the
  rule (\ref{item:D_disj}.--\ref{item:D_anon}.\ and~\ref{x2}). Indeed,
  condition~\ref{item:full_dia}.\ holds by the definition of type-0
  proofs.  As no disjunctive diamond is involved in a type-0 proof,
  all involved clauses are singletons.
  Hence,~\ref{item:D_disj}.,~\ref{item:D_anon}., and~\ref{x2}.\
  directly follow from the observations above.  The type-0 proof at
  hand thus induces a match of rule~$\rs$ to a subset of~$\Xi$; the
  conclusion of this rule match having non-empty extension
  under~$\tau$ means precisely that the resolution proof does not
  produce a blatantly inconsistent clause.

  \subparagraph*{Type-1 proofs} Those consist in successively
  resolving all literals of a single clause of the form~$d_{j_0}$
  where $C_{j_0}\cup H_{j_0}\neq \Agents$ with suitable singleton
  clauses, of the form either~$b^p_j$ or~$d_k$ where
  $C_k\cup H_k=\Agents$. We will refer to these resolution steps as
  `resolving into $d_{j_0}$', although of course $d_{j_0}$ will have
  been modified by previous resolution steps as described above. To
  match the notation of rule $\rc$, we rename
  $\cldiamond{C_{j_0},R_{H_{j_0}}}$ into $\cldiamond{E,Q_K}b$ (so that
  all the $\cldiamond{C_j,R_{H_j}}c_j$ that remain have
  $C_j\cup H_j= \Agents$ and hence $|R_{H_j}|=1$).  The literals
  in~$d_{j_0}$ are then indexed over $q\in Q_K$. Let $\inda_q$ be the
  set of all~$j$ such that for some $p\in P_j$, \(b^p_j\) is resolved
  into \(d_{j_0}\) at the literal for \(q\), and put
  $G=\bigcup_{q\in Q_K, j\in\inda_q}G_j$; similarly, let~$\indb_q$ be
  the set of all~$j$ such that \(d_j\) (a singleton clause) is
  resolved into \(d_{j_0}\) at the literal for \(q\), and
  put~$H=\bigcup_{q\in Q_K, j\in\indb_q}H_j$.  Notice that two clauses
  resolve only if whenever they both assign a constant (either a
  Skolem constant or an explicit move) to a certain agent, then the
  constant is the same in both clauses; this implies
  condition~\ref{x2}.  Conditions~\ref{item:D_disj}.\
  and~\ref{item:D_anon}.\ are established as in the type-0 case,
  condition~\ref{item:full_dia}.\ is ensured by the above renaming,
  and the remaining conditions follow directly from the above
  observations.  The type-1 proof at hand thus induces a match of
  rule~$\rc$ to a subset of~$\Xi$; a conclusion of this rule match
  having non-empty extension under~$\tau$ means precisely that the
  resolution proof does not produce a blatantly inconsistent clause.

  \subparagraph*{Finitely many moves} As indicated above, we obtain a
  model with finitely many moves by constructing a finite
  $\CT$-equationally complete model~$\omod$, where~$\CT$ is the
  unification closure of the tuples of terms occurring in~$\phi$. This
  construction is essentially the same as for the \ac{AMC}, up to the
  presence of additional constant symbols, viz.\ the explicit strategies
  occurring in~$\phi$. These constants can be treated exactly like the
  Skolem constants already present in the proof of~\cref{thm:cl-complete}.
  The full proof is available in~\cref{sec:proof-one-step}.
\end{proof}
Since the rules $\rc$, $\rs$ are algorithmically
sufficiently harmless, our main result follows from \cref{thm:soundness-ess,thm:cldes-complete} by~\cref{thm:coalg-mu}:
\begin{theorem}
  Satisfiability checking for the \ac{AMCDES} is \ET-complete.
\end{theorem}



\section{Conclusions}\label{sec:conclusions}

We have introduced the \acf{AMCDES}, which extends
\acf{ATLES}~\cite{WaltherEA07} with fixpoint operators and disjunction
over explicit strategies of opposing agents in non-grand
modalities. We have employed methods from coalgebraic logic to show
that model checking with fixed interpretation of explicit strategies
is in $\QP$ as well as in $\NP\cap\coNP$, and in $\NP$ with open
interpretation of strategies, and moreover that satisfiability
checking is in $\ET$.

The coalgebraic treatment in fact implies a whole range of additional
results, e.g.\ reasoning in the next-step fragment of the logic
extended with nominals ($\ET$ with global axioms, and $\PS$
without)~\cite{SchroderEA09,MyersEA09,GoreEA10}; cut-free sequent
systems for the next-step fragment~\cite{PattinsonSchroder10}; and
completeness of a Kozen-Park axiomatization for flat (i.e.\ single-variable)
fragments of the \ac{AMCDES}, \eg \ac{ATL} with disjunctive explicit
strategies~\cite{SchroderVenema18}. A special case of the latter
result is completeness of \ac{ATLES} as proved already
in~Walther et al.~\cite{WaltherEA07}.  


In ongoing work we are extending our axiomatization and complexity
results to allow strategy disjunction also in grand coalition
modalities. 
A natural but more challenging further extension would be to add
negative strategies prohibiting moves for some agents as suggested
by Herzig et al.~\cite{Herzig_2013}.  



\bibliography{cl-named}
\clearpage
\appendix

\section{Appendix: AMCDES Model Checking Details}\label{sec:append-amcd-model}
\subsection*{Summary of Results on Coalgebraic Model Checking}
Given a functor~$F$, we assume a representation of the elements of
$FX$, for finite~$X$, as strings over some alphabet. Specifically, we
represent elements of $({(k_j)}_{j\in \Agents},f)\in \Gp X$ as
tabulations of~$f$.

Model checking results~\cite{DBLP:conf/concur/HausmannS19} for the
full coalgebraic \MC require only very simple properties of the
predicate liftings:
\begin{defn}
  The \emph{one-step satisfaction problem} is to determine, given a
  finite set $X$, $Y\subseteq X$, \(\hearts \in \Lambda\), and
  $t\in\F X$, whether $t\in\Sem{\hearts}_X(Y)$.
\end{defn}
\begin{theorem}[Model checking via one-step satisfaction
  {\cite[Theorem~11]{DBLP:conf/concur/HausmannS19}}]\label{thm:one-step-mc}
  If the one-step satisfaction problem is in $\PTIME$, then the model
  checking problem for the coalgebraic \MC over this logic is in
  \(\NP \cap \coNP\).
\end{theorem}
The proof of this upper bound is via parity games,
specifically by noting that C\^irstea et al.'s \emph{evaluation
  games}~\cite{CirsteaEA11b} are exponentially large but have only
polynomially many $\mathsf{Eloise}$-nodes, so that winning strategies
for $\mathsf{Eloise}$ can be guessed and verified in
(nondeterministic) polynomial time.

On the other hand, to obtain a model checking algorithm in \(\QP\)
(deterministic quasipolynomial time $2^{\mathcal{O}((\log n)^k)}$ for
some~$k$; a complexity class not currently known to be comparable
with~$\NP$) we need to show that we can design suitable one-step
satisfaction arenas for use in model checking games (we use standard
terminology for games, e.g.~\cite{lncs2500}):

\begin{defn}
  A \defemph{one-step satisfaction arena}~$\mathsf{A}$ for a set $X$,
  a modality $\hearts \in \Lambda$, and $t\in\F X$ is an acyclic arena
  for games with two players $\mathsf{Eloise}$ and $\mathsf{Abelard}$
  (recall that an arena is like a game in that it specifies nodes,
  each assigned to one of the players, and allowed moves between nodes
  but does not include a winning condition; acyclicity refers to the
  move relation), with a single initial node, with~$X$ as the set of
  terminal nodes, and with additional inner nodes. A \defemph{one-step
    game} on~$\mathsf{A}$ additionally specifies a winning condition
  in the shape of a subset~$Y$ of the terminal nodes; then,
  $\mathsf{Eloise}$ wins plays that either get stuck at an
  inner $\mathsf{Abelard}$ node without successors or terminate in a
  node in~$Y$. We say that~$\mathsf{A}$ is \emph{sound and
  complete} if for every $Y\subseteq X$, $\mathsf{Eloise}$ wins
  (the initial node of) the one-step game on~$\mathsf{A}$ with winning
  condition~$Y$ iff $t\in\Sem{\hearts}_X(Y)$.
\end{defn}
\begin{theorem}[Model checking via one-step games
  {\cite[Corollary~18]{DBLP:conf/concur/HausmannS19}}]\label{thm:one-step-arena-mc}
  If for every set~$X$, $\hearts \in \Lambda$, and~$t \in \F X$, there
  is a sound and complete one-step satisfaction arena with
  polynomially many inner nodes in the representation size
  of~$\hearts$ and~$t$, then the model checking problem for the \MC
  over this logic is in \(\QP\).
\end{theorem}
The model checking procedure underlying this theorem is to
construct a polynomial-size model checking parity game using one-step
games as building blocks; by well-known recent advances in parity game
solving~\cite{CaludeEA17}, these games can be solved in quasipolynomial
time.

\subsection*{Proof of \cref{thm:os-mc}}
\begin{proof}
  The one-step satisfaction problem for the \ac{AMCDES}
  is to check whether $((k_j),f,\interpret)\in\Sem{\clbox{C,O}}_X(Y)$ can be decided in \(\PT\) for given
  $C,O$, $Y\subseteq X$, and a one-step game with explicit strategies
  $((k_j),f,\interpret)\in\GES X$. This can be done by iterating over joint
  moves of~$C$ in an outer loop and over joint moves of~$\overline C$
  in an inner loop. Since~$f$ needs to tabulate the outcomes of all
  joint moves of~$\Agents$, both loops have at most linearly many (in
  the size of~$f$) iterations per invocation, making for a quadratic
  overall number of iterations of the inner loop, and hence polynomial
  run time.

  \begin{algorithm}
    \caption{One-step Satisfaction Algorithm}
    \label{alg:osa}
    \DontPrintSemicolon
    \For{\(m_C \leftarrow \moves{_{C}}\)}{
    \(x \defeq \top\)\;
    \For{\(o \leftarrow O, m_{\bar{C}} \leftarrow \moves{_{\Agents\setminus C \setminus \Ag(O)}}\)}{
    \lIf{\(f(m_c, m_{\bar{C}}, \interpret[o]) \notin Y\)}{\(x \defeq \bot\)}
    }
    \lIf{\(x\)}{return \(\top\)}
    }
    return \(\bot\)
  \end{algorithm}
  By \cref{thm:one-step-mc}, we thus obtain the
  \(\NP \cap \coNP\) bound for the fixed case. The $\NP$ bound for the
  open case follows by guessing history-free strategies.

  For the \QP bound, we use \cref{thm:one-step-arena-mc} and adapt the
  one-step satisfaction arenas for the \ac{AMC} \cite[Example
  15.5]{DBLP:conf/concur/HausmannS19} to obtain small one-step
  satisfaction arenas for the \ac{AMCDES}:

  The one-step satisfaction arena \(A_{\clbox{C,O},w} = (V_{\clbox{C,O}, w}, E_{\clbox{C,O}, w})\) for $X$, $\clbox{C,O}$, and a one-step game $((k_j),f,\interpret)\in\GES X$ for
disjoint \(C,D \subseteq \Agents\), \(O \subseteq \prod_{a \in D}M_a\)
is constructed as follows. The node set~$V_{\clbox{C,O}, w}$ consists
of an initial node \((\clbox{C,O}, w)\) belonging to \(\mathsf{Eloise}\), and
additionally a set of inner nodes
\(I_{\clbox{C,O},w} \defeq \moves{_C}\) belonging to \(\mathsf{Abelard}\) i.e.\ one node for each joint move of \(C\). The set
$E_{\clbox{C,O},w}(x)$ of moves available at a node~$x$ is
\begin{align*}
  E_{\clbox{C,O},w}(x) = \begin{cases}
    I_{\clbox{C,O},w} \text{ if } x = (\clbox{C,O},w)\\
    \{f(x, m_{\bar{C}},o) \mid m_{\bar{C}} \in \moves{_{\,\other{C \cup D}}}, o \in \interpret[O]\}
  \end{cases}
\end{align*}
It is easy to see that the size of the arena is thus
  linear in the tabulation size of~$f$. 
The soundness and completeness of the resulting one-step satisfaction
game stems from the fact that the moves of \(\mathsf{Eloise}\) and \(\mathsf{Abelard}\)
essentially construct the witnessing moves from the original game. 
\end{proof}

\section{Appendix: Omitted Proofs and Further Details}

\subsection*{Proof of~\cref{thm:res-completeness}}\label{sec:proof-thm:r-comp}
\subparagraph*{Soundness}
  It suffices to show that the rule \((SR)\) is sound.  Let
  \(\Gamma,A_1(\bar{t_1}),\dots,A_n(\bar{t_n})\) and \(B(\bar{u}),\Delta\) be two clauses such
  that $\bar t_1,\dots,\bar t_n$, and $\bar u$ are unifiable, and let
  \(\sigma = mgu(\bar{t_1},\dots,\bar{t_n}, \bar{u})\). Let
  \(\omod = ((S_j)_{j \in N}, f, W, \Sem{-})\) be an outcome model
  satisfying both \(\Gamma,A_1(\bar{t_1}),\dots,A_n(\bar{t_n})\) and
  \(B(\bar{u}),\Delta\). Let~$\eta$ be a valuation such
  that~\(\omod,\eta\not\models\Gamma\sigma,\Delta\sigma\); we have to
  show \(\omod, \eta \models ((\bigcup_{i = 1}^nA_i) \cap B)(\bar{u}\sigma)\).  By the
  evident substitution lemma, \(\omod,\eta_\sigma\not\models\Gamma,\Delta\)
  where \(\eta_\sigma(x) = \Sem{\sigma(x)}\eta\) for all \(x\); hence
  necessarily \(\omod, \eta_\sigma \models A_1(\bar{t_1}), \dots, A_n(\bar{t_n})\), and
  \(\omod, \eta_\sigma \models B(\bar{u})\).  Again by the
  substitution lemma, \(\omod, \eta \models A_1(\bar{t_1}\sigma), \dots, A_n(\bar{t_n}\sigma)\), and
  \(\omod, \eta \models B(\bar{u}\sigma)\). Since
  \(\bar{t_1}\sigma = \dots = \bar{t_n}\sigma = \bar{u}\sigma\), our goal
  \(\omod, \eta \models ((\bigcup_{i = 1}^nA_i) \cap B)(\bar{t}\sigma)\) follows by the
  semantics of literals.
\subparagraph*{Completeness}
The completeness proof for the propositional variant
proceeds via \emph{maximally consistent clause sets}, defined in the
expected way. By Zorn's lemma, 
we have

\begin{lemma}[Lindenbaum lemma for set-valued propositional resolution]\label{thm:res-lind}
  Every consistent clause set in set-valued propositional logic is
  contained in a maximally consistent set.
\end{lemma}
Moreover, we have the following set of Hintikka properties:
\begin{lemma}[Hintikka lemma for set-valued propositional resolution]
  Let $\phi$ be a maximally consistent clause set in set-valued
  propositional logic. Then
  \begin{enumerate}
  \item\label{item:comma} A clause $\Gamma,\Delta$ is in~$\phi$ iff
    $\Gamma\in\phi$ or $\Delta\in\phi$.
  \item\label{item:union} A clause $\Gamma,(A\cup B)(y)$ is in $\phi$
    iff one of $\Gamma,A(y)$ and $\Gamma,B(y)$ is in~$\phi$.
  \item\label{item:top} For every $y\in Y$, $W(y)\in\phi$.
  \end{enumerate}
\end{lemma}
\begin{proof}
  \emph{\ref{item:comma}, `if':} Assume w.l.o.g.\ that
  $\Gamma\in\phi$. By maximality, it suffices to show that
  $\phi\cup\{\Gamma,\Delta\}$ remains consistent. So assume that a
  blatantly inconsistent clause can be derived from
  $\phi\cup\{\Gamma,\Delta\}$. Then by removing literals from the
  clauses in this derivation, we obtain a derivation of a blatantly
  inconsistent clause from $\phi\cup\{\Gamma\}$, contradiction.

  \emph{\ref{item:comma}, `only if':} By maximality, it suffices to
  show that one of $\phi\cup\{\Gamma\}$ and $\phi\cup\{\Delta\}$ is
  consistent. Assume the contrary. Then one can derive a blatantly
  inconsistent clause $\Gamma'$ from
  $\phi\cup\{\Gamma\}$. Adding~$\Delta$ to all clauses in the
  derivation (that is, to the original~$\Gamma$ and then to all
  clauses newly produced by the resolution rule), we obtain a
  derivation of $\Gamma',\Delta$ from
  $\phi\cup\{\Gamma,\Delta\}$. Similarly, we have a derivation of a
  blatantly inconsistent clause~$\Delta'$ from $\phi\cup\{\Delta\}$,
  from which we obtain a derivation of $\Gamma',\Delta'$ from
  $\phi\cup\{\Gamma',\Delta\}$. Chaining the two derivations, we
  obtain a derivation of the blatantly inconsistent clause
  $\Gamma',\Delta'$ from $\phi\cup\{\Gamma,\Delta\}$, contradiction.

  \emph{\ref{item:union}, `if':} Assume w.l.o.g.\ that $\Gamma,A(y)$
  is in~$\phi$. By maximality, it suffices to show that
  $\phi \cup \{\Gamma,(A\cup B)(y)\}$ is consistent. Assume the contrary, i.e.\ we
  can derive a blatantly inconsistent clause from
  $\Gamma,(A\cup B)(y)$.  Tracing $(A\cup B)(y)$ through the
  derivation in the obvious sense (with $A\cup B$ possibly transformed
  into strictly smaller subsets by the resolution rule) and
  intersecting with $A$ at each occurrence, we obtain a derivation of
  a blatantly inconsistent clause from $\phi\cup\{\Gamma,A(y)\}=\phi$,
  contradiction.

  \emph{\ref{item:union}, `only if':} By contraposition, again using
  maximality: assume that both $\phi\cup\{\Gamma,A(y)\}$ and
  $\phi\cup\{\Gamma,B(y)\}$ are inconsistent; we have to show that
  $\phi\cup\{\Gamma,(A\cup B)(y)\}$ is inconsistent. By assumption, we can
  derive from $\phi\cup\{\Gamma,A(y)\}$ a blatantly inconsistent
  clause, necessarily of the form $\Gamma',\emptyset(y)$ (since no
  $y\in Y$ can be made to disappear by the resolution rule). Tracing
  $A(y)$ through the derivation and taking unions with $B$ at each
  occurrence, we obtain a derivation of $\Gamma',B(y)$ from
  $\phi\cup\{\Gamma,(A\cup B)(y)\}$. Similarly, we can derive a
  blatantly inconsistent clause from
  $\phi\cup\{\Gamma,B(y)\}$. Replacing literals $C(z)$ with
  $\emptyset(z)$ and adding new literals of the form $\emptyset(z)$,
  we obtain a derivation of a blatantly inconsistent clause $\Theta$
  from $\phi\cup\{\Gamma',B(y)\}$. Chaining derivations, we obtain a
  derivation of~$\Theta$ from $\phi\cup\{(A\cup B)(y)\}$, showing the
  required inconsistency.

\emph{\ref{item:top}:} Clear.
\end{proof}
Now fix a maximally consistent clause set~$\phi$, and assume
that~$W$ is finite; we construct a model, \ie a function
$f_\phi:Y\to W$, from~$\phi$ as follows.  For $y\in Y$, we have
$W(y)\in\phi$ by the Hintikka lemma, and then, again by the Hintikka
lemma and by finiteness of~$W$, $\{w_y\}(y)\in\phi$ for
some~$w_y\in W$, which by consistency of~$\phi$ is moreover unique; we
put $f_\phi(y)=w_y$.
\begin{lemma}[Truth lemma for set-valued propositional resolution]\label{lem:res-truth}
  Given a maximally consistent clause set~$\phi$ in set-valued
  propositional logic over a finite set~$W$, the function~$f_\phi$
  constructed above satisfies~$\phi$.
\end{lemma}
\begin{proof}
  Induction over the size of clauses~$\Gamma$, measured as the sum of
  the cardinalities of the subsets of~$W$ occurring in~$\Gamma$. The
  inductive step makes a case distinction over whether there is more
  than one or exactly one literal in~$\Gamma$ (the case of zero
  literals does not occur, as a clause without literals is blatantly
  inconsistent), and then proceeds according to the relevant clause of
  the Hintikka lemma. We are left with the induction base,
  where~$\Gamma$ has the form $\{w\}(y)$; in this case, the claim
  holds by construction of~$f_\phi$.
\end{proof}
In combination with \cref{thm:res-lind}, this proves
completeness of the propositional variant. Completeness for the
first-order variant is then shown via a form of Herbrand theory. We
build a \emph{Herbrand universe} where the moves of each agent~$i$ are
ground terms of sort~$i$. We denote these sets of moves by~$S_i$. A
\emph{ground substitution} replaces variables by ground terms,
respecting sorts. \emph{Ground instances} of literals $A(\bar t)$,
clauses, and clause sets are obtained by applying a ground
substitution.

Now let $\phi$ be a clause set in set-valued first-order logic that is
closed under set-valued first-order resolution and not blatantly
inconsistent; it suffices to show that such~$\phi$ are satisfiable. By
admissibility of $(\mathit{lSR})$ (\cref{lem:lsr-admissible}), we can
assume that~$\phi$ is even closed under~$(\mathit{lSR})$, and hence
closed under set-valued first-order resolution with two-sided
factoring, i.e.\ under the rule $(\mathit{SR}^+)$, as this rule is
derivable from~$(\mathit{lSR})$ (\cref{lem:two-sided-derivable}). We
denote by $I(\phi)$ the set of ground instances of clauses
in~$\phi$. In the same way as admissibility of $(\mathit{lSR})$ under
$(\mathit{SR})$, one shows easily that a liberalized variant
$(\mathit{lSR}^+)$ of $(\mathit{SR}^+)$ where the substitution applied
is only required to be a unifier (rather than an mgu) is admissible
under $(\mathit{SR}^+)$, so we can assume that~$\phi$ is even closed
under $(\mathit{lSR}^+)$. To show that~$\phi$ is satisfiable over the
Herbrand universe, it suffices to establish that $I(\phi)$ is
satisfiable. Clearly, $I(\phi)$ is not blatantly inconsistent. We show
that it is moreover closed under set-valued propositional resolution
(implying that $I(\phi)$ is satisfiable, and hence that~$\phi$ is
satisfiable). A pair of resolvable clauses in~$I(\phi)$ has the form
$\Gamma\theta,(\bigcup_{i = 1}^nA_i)(\bar t_1\theta)$ and
$(\bigcup_{i=1}^mB_i)(\bar u_1\theta),\Delta\theta$ where
$\Gamma,A_1(\bar t_1),\dots,A_n(\bar t_n)$ and
$B_1(\bar u_1),\dots,B_m(\bar u_m),\Delta$ are in~$\phi$, w.l.o.g.\
with disjoint sets of variables, and $\theta$ is a ground substitution
such that
$\bar t_1\theta = \dots = \bar t_n\theta = \bar u_1\theta=\dots=\bar
u_m\theta$. In particular,~$\theta$ is a unifier of
$\bar t_1, \dots, \bar t_n,\bar u_1,\dots,\bar u_m$. 
It follows that the resolvent
$\Gamma\theta,((\bigcup_{i = 1}^nA_i)\cap (\bigcup_{i=1}^m B)(\bar
u\theta),\Delta\theta$ of $\Gamma,A_1(\bar t_1),\dots,A_n(\bar t_n)$
and $B_1(\bar u_1),\dots,B_m(\bar u_m),\Delta$ under
$(\mathit{lSR}^+)$ is in~$\phi$, and hence (since~$\theta$ is already
ground) 
in~$I(\phi)$; but this clause is the propositional resolvent of the
given clauses $\Gamma\theta,(\bigcup_{i = 1}^nA_i)(\bar t_1\theta)$
and $(\bigcup_{i=1}^mB_i)(\bar u_1\theta),\Delta\theta$, so we are
done. \qed

\subsection*{Remarks on One-step Tableau Completeness for the \ac{AMC} (\cref{thm:cl-complete})}\label{sec:remarks-one-step}
  In the proof of \cref{thm:cl-complete}, one could equally well have used previous
  one-step model constructions implicit
  in~van Drimmelen, Goranko, and Schewe~\cite{Drimmelen03,GorankoDrimmelen06,schewe_PhD}; we provide our
  construction for illustration, in preparation for the treatment of
  disjunctive explicit strategies, to which, as far as we can see, the
  previous constructions do not adapt (they do extend to explicit
  strategies without strategy disjunction). We note that the model
  construction becomes much simpler if one excludes the grand
  coalition (as, effectively, in \ac{ATLES}): In the rule $(C)$, the
  literals $\cldiamond{\Agents}c_j$ disappear; in the proof of
  one-step tableau completeness of the arising rule, one can just use
  a single move $\bot$ as witness for all $\cldiamond{C_j}c_j$
  in~$\Xi$ (in the notation of the original proof of \cref{thm:cl-complete}), using
  non-determinism to ensure satisfaction of the
  $\cldiamond{C_j}c_j$. In detail, this is seen as follows.

  As indicated above, in the absence of grand coalition modalities,
  rule $(C)$ specializes to
  \begin{equation*}
     (C^-)\;\frac{\clbox{D_1}a_1,\dots,\clbox{D_\alpha}a_\alpha, \cldiamond{E}b}{a_1,\dots,a_\alpha, b}
  \end{equation*}
  with the same side conditions as~$(C)$. The shorter proof of
  one-step tableau completeness then runs as follows.  Let $\tau$ be a
  $\Pow W$-valuation, and let
  $\Xi=\{\clbox{D_1}a_1,\dots,\clbox{D_\alpha}a_\alpha,
  \cldiamond{C_1}c_1,\dots,\cldiamond{C_\beta}c_\beta\}$ (where
  $D_j\neq \Agents$, $C_j\neq \Agents$ for all~$j$) be such that every
  rule match of $(C^-)$ to $\Xi$ has non-empty conclusion
  under~$\tau$. We have to construct an element of
  $\Sem{\Xi}\tau$. Give every agent moves~$e_j$ for $j=1,\dots,n$
  intended as witnesses for $\clbox{D_j}a_j$, and a single refusal
  move~$\bot$; write (slightly abusively) $e_{D_j}$ for the joint
  move of~$D_j$ that is~$e_j$ in all components. Define a
  non-deterministic outcome function~$f$ by
  \(
    f(m_\Agents)=\textstyle\bigcap_{e_{D_j}\sqsubseteq m_\Agents}\tau(a_j),
  \)
  noting that this set is non-empty thanks to rule~$(CD)$ since for
  $j\neq k$, having both $e_{D_j}\sqsubseteq m_\Agents$ and
  $e_{D_k}\sqsubseteq m_\Agents$ implies $D_j\cap
  D_k=\emptyset$. Then~$f$ clearly satisfies $\clbox{D_j}a_j$
  under~$\tau$. To see that~$f$ also satisfies $\cldiamond{C_j}c_j$,
  let $m_{C_j}$ be a joint move of~$C_j$. Let $m_\Agents$ be the
  joint move of~$\Agents$ extending $m_{C_j}$ by letting all other
  agents pick~$\bot$. We have to show that
  $f(m_\Agents)\cap\tau(c_j)\neq\emptyset$. But this is immediate by
  rule~$(C^-)$, since $e_{D_k}\sqsubseteq m_\Agents$ implies
  $D_k\subseteq C_j$.

  We note further that excluding grand coalition modalities is
  equivalent to making the outcome function non-deterministic: It is
  clear that excluding grand coalition modalities is equivalent to
  always taking the set of agents to consist of the
  agents~$\Agents_\phi$ mentioned in the target formula~$\phi$ and one
  extra agent~$*$ (convert models with larger set~$C$ of additional
  agents into one with only~$*$ by taking the previous joint moves
  of~$C$ to be the moves of~$*$). Then, note that~$\phi$ is
  satisfiable in a \ac{CGS} with set~$\Agents_\phi\cup\{*\}$ of agents
  iff~$\phi$ is satisfiable in a \emph{non-deterministic \ac{CGS}}
  with set~$\Agents=\Agents_\phi$ of agents, where a non-deterministic
  \ac{CGS} is defined like a \ac{CGS} except that the outcome function
  $f_q$ at a state~$q$ returns a non-empty set of possible post-states
  rather than just a single post-state. Over such a non-deterministic
  \ac{CGS}, a formula $\clbox{C}\psi$ is satisfied at a state~$q$
  if~$C$ has a joint move $m_C$ such that for all joint
  moves~$m_{\other C}$ of~$\other C$, \emph{all} possible
  post-states of~$q$ under the induced joint move of~$\Agents$
  satisfy~$\psi$. A non-deterministic \ac{CGS} with set~$\Agents$ of
  agents is converted into a \ac{CGS} with set~$\Agents\cup\{*\}$ of
  agents by giving~$*$ all states as moves, allowing~$*$ to pick one
  of the possible post-states determined by the other agents (with
  some possible post-state chosen arbitrarily if~$*$ plays a state
  that is not a possible post-state). Conversely, a \ac{CGS}~$S$ with
  set $\Agents\cup\{*\}$ of agents is converted into a
  non-deterministic \ac{CGS} with set~$\Agents$ of agents by taking
  the possible post-states under a joint move $m_\Agents$ of the
  agents in~$\Agents$ to be the set of all post-states of joint moves
  in~$S$ extending~$m_\Agents$. Both conversions clearly preserve
  satisfaction of formulae~$\phi$ mentioning only agents in~$\Agents$.

\subsection*{Proof of One-step Tableau Completeness for the~\ac{AMCDES} (\cref{thm:cldes-complete}) with Finite Sets of Moves}\label{sec:proof-one-step}
\begin{proof}
  Similarly to how the finite moves were achieved in the proof
  of~\cref{thm:cl-complete}, we will colour the moves to simulate the
  effect of the occurs check in unification.  We use the same
  terminology and notation for colours as in the proof
  of~\cref{thm:cl-complete}, and take the colours from the same
  Abelian group~\(U\).  Let \(\phi\) be the clause set constructed in
  the ongoing proof as shown in the main part of the paper.  Now, all
  agents receive (for simplicity) the same moves, namely
  \begin{itemize}
  \item moves $(e^j,0)$ for $j=\rng{\alpha}$, intended as witnesses
    for the moves of the agents in \(D_j\) in $\clbox{D_j, P_{G_j}}a_j$,
  \item moves \((p,0)\) for \(j=\rng{\alpha}\), \(p \in P_{G_j}\) witnessing explicit moves from $\clbox{D_j, P_{G_j}}a_j$,
  \item moves \((r,0)\) for \(j=\rng{\beta}\), \(r \in R_{H_j}\) witnessing explicit moves from $\cldiamond{C_j,R_{H_j}}c_j$, and
  \item moves $(g^j,u)$ for $j=1,\dots,\beta$ and $u\in U$, intended
    as witnesses for $\cldiamond{C_j,R_{H_j}}c_j$.
  \end{itemize}
  Let \(\CT\)  be the unification closure of all argument terms occuring in clauses in \(\phi\). 
  All tuples in~\(\CT\) have the shape \((x_A, e_B, p_C, r_D, g_{\,\other{A \cup B \cup C \cup D}}(x_A,e_B, p_C, r_D))\) where the~\(x_A\) are variables; the~\(e_B\) are Skolem constants and \(p_C\), \(r_D\) are constants for named moves, from possibly different boxes and diamonds; and the \(g_{\,\other{A \cup B \cup C \cup D}}\) are
  Skolem functions from a single diamond, as Skolem functions from multiple diamonds
  do not occur together in the starting terms and such occurrences are not introduced during
  unification due to the occurs check.

  The (finite) model \(\omod\)
  is then defined over coloured moves.  Skolem constants $e^j$ are
  interpreted as $(e^j,0)$, explicit strategies \(r\) and \(p\) are interpreted as \((r,0)\) and \((p,0)\), and Skolem functions \(g_i^j\) for
  $i\in \other{C_j}$ are interpreted as mapping a joint move
  $m_{C_j}$ of~$C_j$ to $(g^j,u_j-\col(m_{C_j}))$ if $i$ is the least
  element of~$\other{C_j}$, and to~$(g^j,0)$ otherwise, thus
  ensuring that $\col(m_{C_j},g^j(m_{C_j}))=u_j$.
  It remains to show that~$\omod$ is $\CT$-equationally complete, obtaining
  by \cref{thm:cmpl-model-satisf} and consistency of~$\phi$ under
  set-valued first-order resolution that~$\phi$ is satisfiable
  over~$\omod$.
  Indeed, observing that the symbols for explicit strategies represent constants in the unification process and are translated exactly like the Skolem constants, we can treat them as part of \(e_B\)  and proceed in the same way as in~\cref{thm:cl-complete}.
\end{proof}

\end{document}